%% file: main.tex
\documentclass[11pt]{article}
\usepackage[english]{babel}
\usepackage[margin=1in]{geometry} 
\usepackage{amsmath,amsthm,amssymb, graphicx, multicol, array}
\usepackage{dsfont}
\usepackage{bbm}
\usepackage{hyperref}
\usepackage{enumerate}
\usepackage{xcolor}
\usepackage{verbatim}
\usepackage[capitalize]{cleveref}
\usepackage{todonotes}
\usepackage{microtype}

\newcommand{\vN}{\mathbf{N}}

\newcommand{\cE}{\mathcal{E}}

\newcommand{\bE}{\mathbb{E}}
\newcommand{\R}{\mathbb{R}}

\newcommand{\pr}{\Pr}
\newcommand{\E}{\bE}      

\newcommand{\TV}{\mathop{\bf TV\/}}

\DeclareMathOperator{\diag}{diag}

\newtheorem{theorem}{Theorem}[section]
\newtheorem{lemma}[theorem]{Lemma}
\newtheorem{corollary}[theorem]{Corollary}
\newtheorem{proposition}[theorem]{Proposition}

\theoremstyle{definition}
\newtheorem{problem}{Problem}
\newtheorem{defn}{Definition}
\newtheorem{remark}{Remark}
\newtheorem{example}{Example}

\newcommand{\gnote}[1]{\textcolor{red}{[Govind: #1]}}
\newcommand{\ynote}[1]{\textcolor{blue}{[Youn: #1]}}

\title{How Many Subpopulations is Too Many? \\ Exponential Lower Bounds for Inferring Population Histories}
\author{
Younhun Kim\footnote{Massachusetts Institute of Technology. Department of Mathematics. Email: \url{younhun@mit.edu}.} \and 
Frederic Koehler\footnote{Massachusetts Institute of Technology. Department of Mathematics. Email: \url{fkoehler@mit.edu}. Research is partially supported by NSF Large CCF-1565235 and Ankur Moitra's David and Lucile Packard Fellowship.} \and 
Ankur Moitra\footnote{Massachusetts Institute of Technology. Department of Mathematics and CSAIL. Email: \url{moitra@mit.edu}. This work was supported in part by NSF CAREER Award CCF-1453261, NSF Large CCF-1565235, a David and Lucile Packard Fellowship, an Alfred P. Sloan Fellowship, and an ONR Young Investigator Award.} \and 
Elchanan Mossel\footnote{Massachusetts Institute of Technology. Department of Mathematics and IDSS. Email: \url{elmos@mit.edu}. Partially supported by awards ONR N00014-16-1-2227, NSF CCF1665252
and DMS-1737944.}  \and 
Govind Ramnarayan\footnote{Massachusetts Institute of Technology. CSAIL. Email: \url{govind@mit.edu}. Partially supported by awards NSF CCF1665252 and DMS-1737944.}}
\date{}

\begin{document}
\begin{titlepage}
\maketitle
\thispagestyle{empty}
\begin{abstract}
Reconstruction of population histories is a central problem in population genetics. Existing coalescent-based methods, like the seminal work of Li and Durbin (\emph{Nature}, 2011), attempt to solve this problem using sequence data but have no rigorous guarantees. 
Determining the amount of data needed to \emph{correctly reconstruct} population histories is a major challenge.
Using a variety of tools from information theory, the theory of extremal polynomials, and approximation theory, we prove new sharp information-theoretic lower bounds on the problem of reconstructing \emph{population structure} \----
the history of multiple subpopulations that merge, split and change sizes over time. Our lower bounds are {\em exponential} in the number of subpopulations, even when reconstructing recent histories. 
 We demonstrate the sharpness of our lower bounds by providing algorithms for distinguishing and learning population histories with matching dependence on the number of subpopulations. Along the way and of independent interest, we essentially
 determine the optimal number of samples needed to learn an exponential mixture distribution information-theoretically, proving the upper bound by analyzing natural (and efficient) algorithms for this problem.
\end{abstract}

\end{titlepage}

\section{Introduction} \label{sec:intro}

\input{intro-new-2}

\subsection{Modeling Assumptions}

Our results will apply under the following assumptions: 
(1) individuals are haploids\footnote{Alternatively, diploids whose phasing is provided.},
(2) the genome can be divided into \emph{known} allelic blocks that are inherited independently and
(3) for each pair of blocks, we are given the exact coalescence time.
Indeed, in practice, one must start with sequenced genomes \---- and in the context of recovering events in human history, (potentially unphased) genotypes of diploid individuals. 
The problem of recovering coalescence times from sequences provides a major challenge and often requires one to either know the population history beforehand, or leverage simultaneous recovery of history and coalescence times using various joint models that enable probabilistic inference.

But since the main message of our paper is a {\em lower bound} on the number of exact pairwise coalescent samples needed to recover population history, in practice it would only be harder. Even in our idealized setting, handling $7$ or $8$ subpopulations already requires more data than one could reasonably be assumed to possess. 
Thus, our work provides a rather direct challenge to empirical work in the area: Either results with $7$ or $8$ subpopulations are not to be trusted or there must be some biological reason why the types of population histories that arise in our lower bounds, that are information-theoretically impossible to distinguish from each other using too few samples, can be ruled out.





\input{coalescent.tex}

\subsection{Our Results}
\label{sec:results}
The main theoretical contribution of this work is an essentially tight bound on the sample complexity of learning population history in the multiple-subpopulation model.
In particular, we show sample complexity lower bounds which are \emph{exponential in the number of subpopulations $k$}.
Here is an organized summary of our results: 
\begin{itemize}
\item First, we show a two-way relationship between the problem of learning a population history (in our simplified model) and the problem of learning a mixture of exponentials. Recall that when the effective subpopulation sizes are all constant, the distribution of coalescence times follows \cref{eq:coalescent-mixture} and thus is equivalent to learning the parameters $p_t$ and $\lambda_t$ in a mixture of exponentials. Conversely, we show how to use an algorithm for learning mixtures of exponentials to reconstruct the entire population history by locating the intervals where there are no genetic events and then learning the associated parameters in each, separately. (Section~\ref{sec:reductions} with details in \cref{sec:proof-equivalence} and \cref{sec:population-history}.)

\item \emph{(Main Result)} Using this equivalence, we show an information-theoretic lower bound on the sample complexity that applies regardless of what algorithm is being used. In particular, we construct a pair of population histories that have different parameters but which require $\Omega((1/\Delta)^{4k})$ samples in order to tell apart. This lower bound is \emph{exponential in the number of subpopulations $k$}. 
Here, $\Delta \le 1/k$ is the smallest gap between any pair of the $\lambda_t$'s. 
The proof of this result combines tools spanning information theory, extremal polynomials, and approximation theory. (\cref{sec:info-overview} with details in \cref{sec:info-lower-bound}.)

\item In the \emph{hypothesis testing} setting where we are given a pair of population histories that we would like to use coalescence statistics to distinguish between, we give an algorithm that succeeds with only $O((1/\Delta)^{4k})$ samples. The key to this result is a powerful tool from analysis, the Nazarov-Tur\'an Lemma \cite{nazarov1993local} which lower bounds the maximum absolute value of a sum of exponentials on a given interval in terms of various parameters. This result matches our lower bounds, thus resolving the sample complexity of hypothesis testing up to constant factors. (\cref{sec:results-nt} with details in \cref{sec:nazarov-turan})

\item In the \emph{parameter learning} setting when we want to directly estimate population history from coalescence times, we give an efficient algorithm which provably learns the parameters of a (possibly truncated) mixture of exponentials given only $O((1/\Delta)^{6k})$ samples. 
We accomplish this by analyzing the \emph{Matrix Pencil Method} \cite{MPM}, a classical tool from signal processing, in the real-exponent setting.
(Section~\ref{sec:matrix-pencil} with details in \cref{sec:matrix-pencil-analysis}.)

\item Finally, we demonstrate using simulated data that our sample complexity lower bounds really do place serious limitations on what can be done in practice. From our plots it is easy to see that the sample complexity grows exponentially in the number of subpopulations even in the optimistic case where the separation $\Delta = 1/k$ which minimizes our lower bounds. In particular, the number of samples we would need very quickly exceeds the number of functionally relevant genes (on the order of $10^{4}$) and even the number of SNPs available in the human genome (on the order of $10^{7}$).
In fact, through a direct numerical analysis of our chosen instances, we can give even stronger sample-complexity lower bounds (\cref{sec:reconstruction-sims}, with details in Appendix~\ref{sec:simulation-math}). 
\end{itemize}
\textbf{Discussion of Results: }
In summary, this work highlights some of the fundamental difficulties of reconstructing population histories from pairwise coalescence data.
Even for recent histories, the lower bounds grow exponentially in the number of subpopulations.
Empirically, and in the absence of provable guarantees, and even with much noisier data than we are assuming, many works suggest that it is possible to reconstruct population histories with as many as nine subpopulations. While testing out heuristics on real data and assessing the biological plausibility of what they find is important, so too is delineating sharp theoretical limitations. Thus we believe that our work is an important contribution to the discussion on reconstructing population histories. It points to the need for the methods that are applied in practice to be able to justify why their findings ought to be believed. Moreover they need to somehow preclude the types of population histories that arise in our lower bounds and are genuinely impossible to distinguish between given the finite amount of data we have access to. 


\input{related-works}

\section{Theoretical Discussion}

\subsection{Reductions between mixtures of exponentials and population history}\label{sec:reductions}
In the rest of our theoretical analysis, we will focus on the mixture of exponentials viewpoint of population history. 
To justify this, note that if we can learn truncated mixtures of exponentials, then we can easily learn population history.
Details are given in \cref{sec:population-history} including a concrete algorithm based on our analysis of the Matrix Pencil Method.
Conversely, we observe that an arbitrary mixture of exponentials can be embedded as a sub-mixture of a simple population history with two time periods, so that recovering the population history requires in particular learning the mixture of exponentials.
The following theorem makes this precise; its proof is delegated to \cref{sec:proof-equivalence}.
\begin{theorem}
\label{thm:equivalence}
Let $P$ with $P(T > t) = \sum_{i = 1}^k p_i e^{-\lambda_i t}$ be a the distribution of an arbitrary mixture of $k$ exponentials (over random variable $T$) with all $\lambda_i > 0$. Then for any $t_0 > 0$, there exists a two-period population history with $k$ populations which induces a distribution $Q$ on coalescence times such that
\[ Q(T > t + t_0 | T \ne \infty, T > t_0) = P(T > t). \]
\end{theorem}
\begin{remark}\label{rmk:short-modern-era}
By choosing a small value for $t_0$, we ensure that very few coalesence times occur in the more recent period, so that the reconstruction algorithm must rely on the information from the second (less recent) period with our planted mixture of exponentials. 
\end{remark}
Additionally, we provide a more sophisticated version of this reduction which maps two mixtures of exponentials to different population histories simultaneously, while preserving statistical indistinguishability.
\begin{theorem}
\label{thm:pop-reduction}
Let $P$ with $P(T > t) = \sum_{i=1}^{k} p_i e^{- \lambda_i t}$ and $Q$ with $Q(T > t) = \sum_{j=1}^{\ell} q_j e^{- \mu_j t}$ be arbitrary mixtures of exponentials with all $\lambda_i, \mu_j > 0$. Then for all 
$t_0 > 0$ sufficiently small there exist two distinct 2-period population histories $R$ with $k + 2$ subpopulations and $S$ with $\ell + 2$ subpopulations such that:
\begin{enumerate}
\item $ R[T > t + t_0 | T\neq \infty, T > t_0] = P(T > t)$ and $ S[T > t + t_0 | T\neq \infty, T > t_0] = Q(T > t)$
\item $R[T = t_0] = S[T = t_0]$ and $R[T = \infty] = S[T = \infty]$. 
\end{enumerate}
\end{theorem}
Again, if we take $t_0$ small enough, we ensure that any distinguishing algorithm must rely on information from the second (less recent) period, and hence because the probability of all other events match, must distinguish between the mixtures of exponentials $Q$ and $R$.


\subsection{Guaranteed recovery of exponential mixtures via the Matrix Pencil Method}
\label{sec:matrix-pencil}
Given samples from a hyperexponential distribution
\begin{equation}
	\label{eq:mixture-exponentials}
	\Pr(T > t) = \sum_{i=1}^{k} p_i e^{-\lambda_i t},
\end{equation}
can we learn the parameters $p_1, \ldots, p_k$, $\lambda_1, \ldots, \lambda_k$?
In \cref{sec:reductions}, we established the equivalence between solving this problem and learning population history.
Suppose for now that we are given access to the exact values of probabilities $v_t := \pr(T \geq t)$ for $t \in \mathbb{R}_{\geq 0}$, i.e.
$
	v_t 
    = \sum_{j = 1}^k p_i \alpha_i^t,
$
where $\alpha_i = e^{-\lambda_i}$. 
The \emph{Matrix Pencil Method} is the following linear-algebraic method, originating in the signal processing literature \cite{MPM}, which solves for the parameters $\{ p_i , \lambda_i \}_{i=1}^{k}$:



\begin{enumerate}[1.]
\item Let $A,B$ be $k \times k$ matrices 
where $A_{ij} = v_{i+j-1}$ and $B_{ij} = v_{i+j-2}$.
\item Solve the generalized eigenvalue equation $\det(A - \gamma B) = 0$ for the pair $(A,B)$. The $\gamma$ which solve $\det(A -\gamma B) = 0$ are the $\alpha$'s.
\item Finish by solving for the $p$'s in a linear system of equations $\vec{v} = V\vec{p}$, where $\vec{v} = (v_0, \ldots, v_{k-1})$, $V$ is the $k \times k$ Vandermonde matrix generated by $\alpha_1, \ldots, \alpha_k$ and $\vec{p}$ is the vector of unknowns $(p_1, \ldots, p_k)$.
\end{enumerate}
\noindent
To understand why the algorithm works in the noiseless setting, consider the decomposition $A = VD_p D_\alpha V^T$ and $B = V D_p V^T$ where $V = V_{k}(\alpha_1, \ldots, \alpha_k)$ is the $k \times k$ Vandermonde matrix whose $(i,j)$ entry is $\alpha^{i-1}_j$, $D_\alpha = \diag(\alpha_1, \ldots, \alpha_k)$ and $D_p = \diag(p_1, \ldots, p_k)$. Then it's clear that the $\alpha_i$ are indeed the generalized eigenvalues of the pair $(A,B)$.
However, in our setting, we do not have access to the exact measurements $v_t$, but instead have \emph{noisy} empirical measurements $\tilde{v}_t$; in practice, the output of the MPM can be very sensitive to noise. 

The Matrix Pencil Method is a close cousin of \emph{Prony's Method}~\cite{prony}.
Prior to this work, Feldmann et al. \cite{feldmann1998fitting} considered the strategy of using Prony's Method to \emph{fit} exponential mixtures to general long-tail distributions.
In the upcoming section, we provide a finite-sample guarantee of the MPM in the context of \emph{learning} exponential mixture distributions.
As it turns out (\cref{rmk:mpm-optimality-1} and \cref{rmk:mpm-optimality-2}), this algorithm is nearly optimal in terms of the number of samples required.




\subsubsection{Analysis of MPM under noise}
We now describe our analysis of the MPM in the more realistic setup where the CDF is estimated from sample data.
First note that the model (Equation~\ref{eq:mixture-exponentials}) is statistically unidentifiable if there exist two identical $\lambda$'s.
Indeed, the mixture $\frac{1}{2} e^{-\lambda t} + \frac{1}{2} e^{-\lambda t}$ is exactly same as the single-component model $e^{-\lambda t}$, as is any other re-weighting of the coefficients into arbitrarily many components with exponent $\lambda$.
Therefore it is natural to introduce a \textit{gap parameter} $\Delta := \min_{i \neq j} | \lambda_i - \lambda_j |$ which is required to be nonzero, as in the work on super-resolution (e.g. \cite{candes2013super, moitra14}). 

Without loss of generality, we also assume that:
(1) the components are sorted in decreasing order of exponents, so that $\lambda_1 > \cdots > \lambda_k > 0$, and
(2) time has been re-scaled\footnote{In practice, even if this scaling is unknown, this is easily handled by e.g. trying powers of 2 and picking the best result in CDF distance, for instance $\|F - G\|_\infty = \sup_t |F(t) - G(t)|$.} by a constant factor, so that $\lambda_i \in (0,1)$ for each $i$. 
Now we can state our guarantee for the MPM under noise:
\begin{theorem} \label{thm:matrix-pencil-sample-complexity}
Let $\Delta = \min_{i \neq j} |\lambda_i - \lambda_j|$ and let $p_{\min} = \min_{i} p_i$.
For all $\delta > 0$, there exists
$N_0 = O\left( \frac{ k^{10} }{ p_{\min}^4 } \left( \frac{2e}{\Delta} \right)^{6k} \log \frac{1}{\delta} \right)$
such that, with probability $1 - \delta$, using empirical estimates $\widetilde{v}_0, \ldots, \widetilde{v}_{2k-1}$ from $N \ge N_0$ samples, the matrix pencil method outputs $\{ (\widetilde{\lambda}_j, \widetilde{p}_j) \}_{j=1}^{k}$ satisfying
\begin{equation*}
	| \widetilde{\lambda}_{j} - \lambda_{j} |
    =
    O \left( \frac{k^{3.5}}{p_{\min}^2} \left( \frac{2e}{\Delta} \right)^{2k} \sqrt{\frac{1}{N} \log \frac{1}{\delta} } \right)
    \quad\text{and}\quad
    | \widetilde{p}_j - p_j | 
    =
    O \left( \frac{k^{5}}{p_{\min}^2} \left( \frac{2e}{\Delta} \right)^{3k} \sqrt{\frac{1}{N} \log \frac{1}{\delta} } \right)
\end{equation*}
for all $j$.
\end{theorem}
\begin{remark}\label{rmk:mpm-assumptions}
Letting $\alpha_i$ denote $e^{-\lambda_i}$, we note that we can equivalently focus on learning the $\alpha_i$'s, and that guarantees for recovering $\lambda_i$ and $\alpha_i$ are equivalent up to constants:
$e^{-1}|\alpha_i - \widetilde{\alpha}_i| \le |\lambda_i - \widetilde{\lambda}_i| \le |\alpha_i - \widetilde{\alpha}_i |.$
since $e^{-x}$ is monotone decreasing on $[0,1]$ with derivative lying in $[-1,-1/e]$.
\end{remark}
The full proof of Theorem~\ref{thm:matrix-pencil-sample-complexity} is given in \cref{sec:matrix-pencil-analysis}. As in previous work analyzing the MPM in the \emph{super-resolution} setting with imaginary exponents \cite{moitra14}, we see that the stability of MPM ultimately comes down to analyzing the condition number of the corresponding Vandermonde matrix, which in our case is very well-understood \cite{gautschi1962inverses}.

\subsection{Strong information-theoretic lower bounds}\label{sec:info-overview}
In this section we describe our main results, strong information theoretic
lower bounds establishing the difficulty of learning mixtures of exponentials
(and hence, by our reductions, population histories). 
The full proofs of all results found in this section
are given in \cref{sec:info-lower-bound}.
First, we state a lower bound on learning the exponents $\lambda_j$,
which is an informal restatement of Corollary~\ref{cor:turan-tv}.

\begin{theorem}
\label{thm:turan-sample-complexity-intro}
For any $k > 1$, there exists an infinite family of parameters $a_1,\ldots,a_k,\lambda_1,\ldots,\lambda_k$ and $b_1,\ldots,b_k,\mu_1,\ldots,\mu_k$ parametrized by integers $m > 2(k-1)$ and $\alpha \in (0,\frac{1}{2})$ such that:
\begin{enumerate}
	\item  Each $\lambda_i$ and $\mu_j$ is in $(0,1]$, $\lambda_1 = \mu_1$, and the elements of $\{\lambda_i\}_{i=2}^k \cup \{ \mu_i\}_{i=2}^k$ are all distinct and separated by at least $\Delta = 1/(m+2k)$. Furthermore $\lambda_2,\mu_2 > \alpha/k$.
    \item Let $H_1$ and $H_2$ be hypotheses, under which the random variable $T$ respectively follows the distributions
    \[
	\Pr_{H_1}[T \geq t] = \sum_{i=1}^k a_i e^{-\lambda_i t}
    \quad and \quad
    \Pr_{H_2}[T \geq t] = \sum_{i=1}^k b_i e^{-\mu_i t}.
	\]
    If $N$ samples are observed from either $H_1$ or $H_2$, each with prior probability $1/2$, then the Bayes error rate for any classifier that distinguishes $H_1$ from $H_2$ is at least $\frac{1 - \delta}{2}$, where
    \[ \delta = \frac{\alpha \sqrt{2N}}{2k-3} [\Delta(2k-3)]^{2k-4}. \]
\end{enumerate}
\end{theorem}

\begin{remark}\label{rmk:mpm-optimality-1}
From the square-root dependence of $N$ in 
Theorem~\ref{thm:matrix-pencil-sample-complexity}, the required number of samples $N_0$ has rate $4k$ in the exponent of $\frac{2e}{\Delta}$ if one just wants to learn the $\lambda$'s, and
Theorem~\ref{thm:turan-sample-complexity-intro} confirms that the exponent $4k$ is tight for learning the $\lambda$'s. 
\end{remark}

Next we state an additional information-theoretic lower bound showing that the information-theoretic (minimax) rate is necessarily of the form $\frac{1}{\sqrt{N}} \Delta^{-O(k)}$ up to lower order terms, \emph{even if} all of the $\lambda_i$ are already known and we are only asked to reconstruct the mixing weights $p_j$.

\begin{theorem}
\label{thm:minimax-lower-bound-intro}
Let $m$, $k$ be positive integers such that $m > k > 3$ and let $\Delta = 1/(m + k - 1)$.
There exists a fixed choice of $\lambda_1,\ldots,\lambda_{k}$ which are $\Delta$-separated such that
\begin{equation} \label{eq:minimax-bound}
	\inf_{\hat{p}} \max_{p} \E_{p} \|p - \hat{p}\|_1 
	\ge 
	\frac{1}{4} \min\left( 
    	1,
        \frac{k-3}{\sqrt{2N}} \left( \frac{1}{\Delta(k-3)} \right)^{k-4}
    \right)
\end{equation}
where the max is taken over feasible choices of $p$, and the infimum is taken over possible estimators $\hat{p}$ from $N$ samples of the mixture of exponentials with CDF $F(t) = 1 - \sum_j p_j e^{-\lambda_j t}$.
\end{theorem}


\begin{remark}\label{rmk:mpm-optimality-2}
Recall that in Theorem~\ref{thm:matrix-pencil-sample-complexity}, 
the number of samples needed was exponential in $4k$ when learning just the $\lambda$'s and in $6k$ 
for learning both the $\lambda$'s and the $p$'s.
The exponent of $2k$ in \cref{thm:minimax-lower-bound-intro}
suggests that the discrepancy of $2k$ for MPM in Theorem~\ref{thm:matrix-pencil-sample-complexity} is tight.
\end{remark}



As expected, our lower bounds show that the learning problem becomes harder as $\Delta$ approaches 0.
The ``easiest" case, then, ought to be when $\Delta$ is as large as possible, so that 
the $\lambda_i$ are equally spaced apart in the unit interval. This raises the following question: as $\Delta$ grows, does the sample complexity remains exponential in $k$, 
or is there a phase transition (as in super-resolution \cite{moitra14}) where the problem becomes easier?
In the Appendix, we completely resolve this question: the sample complexity still grows exponentially in $4k$ (\cref{thm:chebyshev-1}) when $\Delta$ is maximally large.

\subsection{A tight upper bound: Nazarov-Tur\'an-based hypothesis testing} \label{sec:results-nt}
As an alternative to the learning problem that the Matrix Pencil Method solves, we also consider the hypothesis testing scenario in which we want to test if the sampled data matches a hypothesized mixture distribution. In this case, we can give guarantees from weaker assumptions and requiring smaller numbers of samples. To state our guarantee, we need the following additional notation: for $P$ a mixture of exponentials, let $p_{\lambda}(P)$ denote the coefficient of $e^{-\lambda t}$, which is 0 if this component is not present in the mixture. We study the following simple-versus-composite hypothesis testing problem using $N$ samples:

\begin{problem}\label{problem:testing2}
Fix $k_0,k_1,\delta,\Delta > 0$ and let $P$ be a known mixture of $k_0$ exponentials. 
\begin{itemize}
\item $H_0$: The sampled data is drawn from $P$.

\item $H_1$: The sampled data is drawn from a different, unknown mixture of at most $k_1$ exponentials $Q$. Let $\nu_1 := \max \{ \lambda : p_{\lambda}(P) > p_{\lambda}(Q) \}$ and $\nu_2 := \max \{ \lambda : p_{\lambda}(Q) > p_{\lambda}(P) \}$. 
We assume that $\min \{ |p_{\nu_1}(P) - p_{\nu_1}(Q)|, |p_{\nu_2}(P) - p_{\nu_2}(Q)| \} \ge \delta$ and $|\nu_1 - \nu_2| \ge \Delta$.
\end{itemize}
\end{problem}
Henceforth, we will refer to $H_0$ as the \emph{null hypothesis} and $H_1$ as the \emph{alternative hypothesis} (note that $H_1$ is a composite hypothesis).
To solve 
this hypothesis testing problem, we propose a finite-sample variant of the Kolmogorov-Smirnov test:
\vspace{-4mm}
\begin{center}
\fbox{\parbox{\textwidth}{
\begin{enumerate}
\item Let $\alpha > 0$ be the significance level.
\item Let $F_N$ be the empirical CDF and let $F$ be the CDF under the null hypothesis $H_0$.
\item Reject $H_0$ if $\sup_t |F_n(t) - F(t)| > \sqrt{\log(2/\alpha)/2N}$.
\end{enumerate}
}}
\end{center}
We show 
that this test comes with a provable finite-sample guarantee.

\begin{theorem}  \label{thm:nt-hypothesis-testing}
Consider the problem setup as in \cref{problem:testing2} 
and fix a significance level $\alpha > 0$. Let $k := \frac{k_0 + k_1}{2}$ and $c_{\Delta} = 8e^2/\min(1/\Delta, 2k - 1)$. Then:
\begin{enumerate}
\item (Type I Error) Under the null hypothesis, the above test rejects $H_0$ with probability at most $\alpha$. 
\item (Type II Error) There exists 
$N_0(\alpha) = O((c_{\Delta}/\Delta)^{4k - 2} \log(2/\alpha)/\delta^2)$
such that if $N \ge N_0$, then
the power of the test at significance level $\alpha$ is at least:
\begin{equation} \label{eq:nt-error}
	\Pr_{Q \in H_1}[\text{Reject } H_0] 
    \ge 
    1 - 2\exp\left(-N\delta^2 (\Delta/c_{\Delta})^{4k - 2}/8\right).
\end{equation}
\end{enumerate}
\end{theorem}
The full proof of Theorem~\ref{thm:nt-hypothesis-testing} is given in \cref{sec:nazarov-turan}. The key step in the proof is a careful application of the celebrated Nazarov-Tur\'an Lemma \cite{nazarov1993local}.

\begin{remark}
This improves upon the Matrix Pencil Method upper bound (\cref{thm:matrix-pencil-sample-complexity}), in terms of the exponent found above $\Delta$ ($\Delta^{-6k}$ versus $\Delta^{-4k}$) and above the mixing weights ($p_{\min}^4$ versus $\delta^2$).
Even when the alternative $Q$ is fixed and known, we see from Theorem~\ref{thm:turan-sample-complexity-intro} that $\Omega((1/\delta^2) (1/\Delta)^{4k})$ many samples are information-theoretically required, which matches \cref{thm:nt-hypothesis-testing}.
\end{remark}

\input{simulations}

\newpage
\bibliographystyle{plain}
\bibliography{bib}

\newpage
\appendix
\section{Derivation of the multiple-subpopulation coalescent model}
\label{sec:model-derivation}
For $b > a > 0$, let $I = [a,b]$ be an interval such that the population structure $\vN$ is constant over $I$.
Then \cref{eq:coalescent-regular}, together with the Markov property of Kingman's coalescent model tells us that the coalescence time $T$ of two randomly sampled individuals in the $i$th sub-population is given, for any $t \in [0,b-a]$, by
\begin{equation}
\label{eq:coalescent-single}
	\Pr\left( T > a + t \;\big|\; \cE_{i,i} \land \{ T > a \} \right) = \exp \left( - \frac{1}{N_i} t \right).
\end{equation}
Here, $\cE_{i,j}$ represents the event where the ancestry of one of the individuals traces back to subpopulation $i$, and the other traces back to $j$.

Let $D$ be the number of subpopulations restricted to the interval $I$.
By the law of total probability, the random variable $T$ satisfies, again for any $t \in [0,b-a]$,
\[
	\Pr(T > a + t \;\big|\; T > a)
	= 
	\sum_{i < j} \Pr(\cE_{i,j} \;\big|\; T > a)
	+ \sum_{i = 1}^{D} \Pr(\cE_{i,i} \;\big|\; T > a) \Pr( T > a+t \;\big|\; \cE_{i,i} \land \{T > a\})
\]
The first summation over $i < j$ uses the fact that $\Pr(T > a + t \;\big|\; \cE_{i,j} \land \{T > a\}) = 1$, via the ``no admixture'' assumption; whenever the two individuals' lineages at time $a$ lie in distinct subpopulations, they do not coalesce anywhere in $I$.
Via \cref{eq:coalescent-single}, the right hand side can be re-written as seen in \cref{eq:coalescent-mixture}, i.e.
\[
	\Pr(T > a + t \;\big|\; T > a)
	=
	\sum_{\ell = 0}^{D} p_{\ell} e^{-\lambda_\ell t}.
\]
\section{Reduction from Mixtures of Exponentials to Population History}
\label{sec:proof-equivalence}
\subsection{Proof of \cref{thm:equivalence}}
\begin{proof}
We consider the following population history:
\begin{itemize}
\item In the (more recent) period $[0,t_0]$ there are $k$ populations and population
$i$ has size $\sqrt{q_i}$, where $q_i$ is the (unique) nonnegative solution to
\[ p_i = q_i e^{-t_0/\sqrt{q_i}}. \]
To see that the solution exists and is unique, observe that the rhs of this equation is a strictly increasing function in $q_i$ which maps $(0,\infty)$ to $(0,\infty)$. 
\item In the (less recent) period $[t_0, \infty)$ each of the $k$ populations changes to
size $1/\lambda_i$. 
\end{itemize}
By construction, the probability that two independently sampled individuals being in the same population is proportional to $q_i$, and conditioned on no coalescence before time $t_0$ this probability
is proportional to $p_i$. Therefore the distribution $Q$ of coalescence times satisfies
\begin{equation} 
\label{eqn:equivalence}
Q(T > t + t_0 | T > t_0) = q_0 + \frac{1}{1 - q_0} \sum_{i = 1}^k p_i e^{-\lambda _i t} 
\end{equation}
where $q_0$ is the probability that the two individuals sampled were in different populations.
\end{proof}

\input{population-reduction}

\section{Analysis of the Matrix Pencil Method} \label{sec:matrix-pencil-analysis}

In this section, $\| A \|$ and $\| A \|_F$ respectively denotes the operator norm and the Frobenius norm of matrices.
For a vector $x$, $\|x\|$ is its Euclidean norm, and more generally $\| x \|_p$ denotes its $\ell_p$-norm.

\subsection{The condition number of Vandermonde matrices}
In 1962, Gautschi \cite{gautschi1962inverses} observed an exact formula
for the $\ell_{\infty} \to \ell_{\infty}$ condition number of a real
Vandermonde matrix, which we now recall:
\begin{defn} \label{def:infty-infty-norm}
The $\ell_{\infty} \to \ell_{\infty}$ norm of a matrix is defined by\footnote{The second equality follows from Holder's inequality (and its equality case, where one takes $x$ to be the appropriate sign vector).}
\[ \|A\|_{\infty \to \infty} :=  \sup_{x : \|x\|_{\infty} = 1} \|A x\|_{\infty} = \max_i \sum_j |A_{ij}|. \]
\end{defn}
\begin{theorem}[\cite{gautschi1962inverses},\cite{gautschi1990stable}]
\label{thm:gautschi}
Suppose $V = V_{n}(\alpha_1, \ldots, \alpha_n)$ such that $\alpha_1,\ldots,\alpha_n$ are all distinct. 
Then
\[ \max_i \prod_{j : j \ne i} \frac{\max(1,|\alpha_j|)}{|\alpha_i - \alpha_j|} \le 
\|V^{-1}\|_{\infty \to \infty} \le \max_{i} \prod_{j : j \ne i} \frac{1 + |\alpha_j|}{|\alpha_j - \alpha_i|} \]
and furthermore equality is attained in the upper bound whenever
$\alpha_j = |\alpha_j| e^{i \theta}$ for $\theta$ independent of $j$.
\end{theorem}
Since we are interested in real-valued Vandermonde matrices with positive entries, the above expression is an exact formula. 
Furthermore, we can relate $\|V^{-1}\|_{\infty \to \infty}$ to the bottom singular value of $V$, i.e. top singular value of $V^{-1}$, because for any matrix $A$ we have 
\[ \sigma_{\max}(A) = \sup_{x : \|x\|_2 = 1} \|A x\|_2 \le \sqrt{n} \sup_{x : \|x\|_2 = 1} \|A x\|_{\infty} \le \sqrt{n} \sup_{x : \|x\|_{\infty} = 1} \|A x\|_{\infty} = \sqrt{n}\|A\|_{\infty \to \infty}  \]
for an upper bound, and for a lower bound we have
\[ \sigma_{\max}(A) = \sup_{x : x \ne 0} \frac{\|A x\|_2}{\|x\|_2} \ge \sup_{x : x \ne 0} \frac{\|A x\|_2}{\sqrt{n}\|x\|_{\infty}} \ge \sup_{x : x \ne 0} \frac{\|A x\|_{\infty}}{\sqrt{n} \|x\|_{\infty}} = \frac{1}{\sqrt{n}} \|A\|_{\infty \to \infty} \]
Hence, by applying these bounds to the matrix $V^{-1}$ and from Theorem~\ref{thm:gautschi}, we see that 
\[ \sigma_{min}(V)^{-1} = \sigma_{max}(V^{-1}) \in \left(\frac{1}{\sqrt{n}} \max_{i} \prod_{j : j \ne i} \frac{1 + |\alpha_j|}{|\alpha_j - \alpha_i|}, \sqrt{n}\max_{i} \prod_{j : j \ne i} \frac{1 + |\alpha_j|}{|\alpha_j - \alpha_i|} \right) \]
In particular, these bounds show that the condition number of a square Vandermonde is exponentially bad in the dimension. 
(Note that these bounds do not contradict Moitra's results \cite{moitra14}, which instead shows that in the Fourier setting, rectangular $V$ can sometimes be well-conditioned -- this is because \cref{thm:gautschi} looks specifically at the conditioning of square Vandermonde matrices.)

Next, we translate this result to a bound in terms of the parameters described in \cref{sec:matrix-pencil}.

\begin{lemma} \label{lem:singular-value-vandermonde}
Let $\alpha_i = e^{-\lambda_i}$ and $1 > \lambda_1 > \cdots > \lambda_k > 0$.
Define $V = V_k(\alpha_1,\ldots,\alpha_k)$, and let $\Delta = \min_{i \neq j} | \lambda_i - \lambda_{j}|$.
Then $\frac{1}{\sigma_{min(V)}} \leq \sqrt{k} \left( \frac{2e}{\Delta} \right)^{k}$ and
$\kappa(V) \leq k^{3/2} \left( \frac{2e}{\Delta} \right)^k$.
\end{lemma}
\begin{proof}
Observe that
\[
	\prod_{j \neq i} \frac{1 + |\alpha_j|}{|\alpha_j - \alpha_i|} \le \left(\frac{2e}{\Delta}\right)^k
\]
because $|\alpha_i| \le 1$ and 
\[ \alpha_i - \alpha_j = e^{-\lambda_i} - e^{-\lambda_j} = \int_{-\lambda_j}^{-\lambda_i} e^{x} dx \ge \Delta e^{-1} \]
for $i > j$.
Thus $\|V^{-1}\|_{\infty \to \infty} \leq \left( \frac{2e}{\Delta} \right)^{k}$, which implies
$\frac{1}{\sigma_{\min}(V)} = \sigma_{\max}(V^{-1}) \leq \sqrt{k} \left( \frac{2e}{\Delta} \right)^{k}$.
We also know that $\sigma_{\max}(V) \le \|V\|_F \le k$, which gives the bound on the condition number.
\end{proof}

\subsection{Matrix perturbation bounds} \label{sec:perturbation-bounds}

In this section, we will establish key lemmas that allow us to prove bounds on how close $\widetilde{\alpha}_j$ are to $\alpha_j$.

\begin{lemma} \label{lem:eval-error-bound-diagonalized-new}
Let $(A,B)$ be a pair of $n \times n$ matrices with generalized eigenvalues $\{ \mu_j \}$, such that $B$ is nonsingular. Take $A = VD_A V^T$ and $B = VD_B V^T$ where $D_A$ and $D_B$ are diagonal matrices and $V$ is an arbitrary, nonsingular $n \times n$ matrix.
Consider the perturbed system $\widetilde{A} = A + E$ and $\widetilde{B} = B + F$ where $E$ and $F$ are symmetric matrices, and let $\{\widetilde{\mu}_j\}$ be its generalized eigenvalues.
Assume further that $\|F\| < \sigma_{min}(D_B)\sigma_{min}(V)^2$.
Then for all $j$,
\begin{equation}
\label{eq:eval-error-bound}
	|\mu_j - \widetilde{\mu}_j|
    \le
     \frac{2k^{3/2}}{\sigma_{\min}(D_B) \sigma_{\min}(V)^2 - \|F\|} \left(\frac{\sigma_{\max}(D_A)}{\sigma_{\min}(D_B)} \|F\| + \|E\|\right)
\end{equation}
\end{lemma}
\begin{proof}
Observe that the generalized eigenvalue problem $Ax = \mu B x$ has the same solutions as the ordinary eigenvalue problem $B^{-1} A x = \mu x$.  Let $E' = V^{-1} E (V^T)^{-1}$ and $F' = V^{-1} F (V^T)^{-1}$ so that 
\[ \tilde{A} = V(D_A + E')V^T, \qquad \tilde{B} = V(D_B + F')V^T. \]
Note that $\|F'\| \le \|V^{-1}\|^2 \|F\| < \sigma_{\min}(D_B)$ by assumption, so $\tilde{B}$ is also invertible, hence the generalized eigenvalues $\{\tilde{\mu}_j\}$ of $(\tilde{A},\tilde{B})$ are just the ordinary eigenvalues of $\tilde{B}^{-1} \tilde{A}$. Since eigenvalues are invariant under change of basis (i.e. similarity transformation), the eigenvalues of $B^{-1} A$ are the same as those of
\[ C := V^T B^{-1} A (V^T)^{-1} = D_B^{-1} D_A \]
and the eigenvalues of $\tilde{B}^{-1} \tilde{A}$ are the same as those of
\[ \tilde{C} := V^T \tilde{B}^{-1} \tilde{A} (V^T)^{-1} = (D_B + F')^{-1}(D_A + E'). \]
Therefore if we let $\mathcal{E} := C - \tilde{C}$, then by Gershgorin's circle theorem\footnote{Here we use that the connected component made of $r$ Gershgorin discs has exactly $r$ eigenvalues, which follows by a standard continuity argument. Therefore the distance an eigenvalue moves under perturbation is at most the sum of the diameters of the discs.}, the Cauchy-Schwarz inequality, and the equivalence of Frobenius and spectral norms, we find
\begin{equation}\label{eqn:mu-j}
 |\mu_j - \tilde{\mu}_j| \le 2\sum_{i,j} |\mathcal{E}_{ij}| \le 2k \|\mathcal{E}\|_F \le 2k^{3/2} \|\mathcal{E}\|
\end{equation}
for all $j$.

It remains to bound $\|\mathcal{E}\|$. 
By the triangle inequality,
\begin{align*}
\|\mathcal{E}\|
= \|D_B^{-1} D_A  - (D_B + F')^{-1}(D_A + E')\|
\le \|D_B^{-1} D_A - (D_B + F')^{-1} D_A\| + \|(D_B + F')^{-1}E'\|.
\end{align*}
To bound the second term on the right hand side, we observe
\[ \|(D_B + F')^{-1}E'\| \le \|(D_B + F')^{-1}\| \|E'\| \le \frac{\|E'\|}{\sigma_{\min}(D_B) - \|F'\|}. \]
To bound the first term, we
observe the following useful matrix identity for $\|M\| < 1$:
\[ (I - M)^{-1} = \sum_{i = 0}^{\infty} M^i = I + M(I - M)^{-1}. \]
This gives
\[ (D_B + F')^{-1} = D_B^{-1} (I + F' D_B^{-1})^{-1} = D_B^{-1} - D_B^{-1} F'(D_B + F')^{-1}. \]
which is valid because $\|F'\| \le \|F \| \|V^{-1}\|^2 < \sigma_{\min}(D_B)$ by assumption.
Therefore
\begin{align*}
\|D_B^{-1} D_A - (D_B + F')^{-1} D_A\| 
= \|D_B^{-1} F'(D_B + F')^{-1} D_A\|
\le \frac{\|D_B^{-1}\| \|F'\| \|D_A\|}{\sigma_{\min}(D_B) - \|F'\|}.
\end{align*}
Combining these two parts completes the bound on $\|\mathcal{E}\|$:
\begin{align*}
\|\mathcal{E}\| 
&\le \frac{1}{\sigma_{\min}(D_B) - \|F'\|} (\|D_B^{-1}\| \|F'\| \|D_A\| + \|E'\|) \\
&\le \frac{\|V^{-1}\|^2}{\sigma_{\min}(D_B) - \|V^{-1}\|^2 \|F\|} (\|D_B^{-1}\| \|F\| \|D_A\| + \|E\|).
\end{align*}
Rewriting the last expression and combining with \eqref{eqn:mu-j} gives the result.
\end{proof}
\subsection{Proof of Theorem~\ref{thm:matrix-pencil-sample-complexity}}
When we put the results of \cref{sec:perturbation-bounds} back into the context of learning mixtures of exponentials, we should think about the perturbation errors $E$ and $F$ as essentially being the same, just offset in the row/column indexing from each other. This is because when we are using the Matrix Pencil Method (\cref{sec:matrix-pencil}) to learn an exponential mixture, the entries of $A$ and $B$ are simply noisy versions of $v_0, v_1, v_2, \ldots, v_{2k-1}$.

We remind the reader of the normalization assumption which restricts $\lambda_1, \ldots, \lambda_k \in [0,1]$.
An application of the previous lemmas yields the following result:

\begin{lemma} \label{lem:exponent-perturbation} 
Consider $\widetilde{v}_0 \ldots, \widetilde{v}_{2k-1}$ as inputs to the Matrix Pencil Method, in place of $v_0, \ldots, v_{2k-1}$.
Let $p_{\min} := \min_i p_i$.
If
\[ 
	\sup_{t \geq 0} |v_t - \widetilde{v}_t| 
    = 
    \frac{\beta p_{\min}}{k^2} \left( \frac{2e}{\Delta} \right)^{-2k}
\]
for some $\beta \in [0,1)$, then for all $j$, 
\[
	|\alpha_j - \widetilde{\alpha}_j|
    \leq
    \frac{ k^{3/2} \left( \frac{1}{p_{\min}} + 1 \right) }{ \frac{1}{\beta} - 1 }.
\]
\end{lemma}
\begin{proof}
The strategy is to apply Lemma~\ref{lem:eval-error-bound-diagonalized-new}.
The setup prescribes $D_A = D_p D_\alpha$, $D_B = D_p$, which gives 
$\sigma_{\min}(D_B) = p_{\min}$ and $\sigma_{\max}(D_A) \leq 1$.
Observe that the hypothesized bound on the estimated $v_t$'s implies
\[
	\max(\|E\|_\infty, \|F\|_\infty) 
    = \max_{t \in \{0,\ldots,2k-1\}} v_{t}
    \leq \frac{\beta p_{\min}}{k^2} \left( \frac{2e}{\Delta} \right)^{-2k}.
\]
This, in turn, gives us bounds on $\|E\|$ (and $\|F\|$), since $\|E\| \leq k\|E\|_\infty$.
In particular, by Lemma~\ref{lem:singular-value-vandermonde}, $\|F\| \leq \sigma_{\min}(D_B) \sigma_{\min}(V)^2$.

This allows us to directly apply Lemma~\ref{lem:eval-error-bound-diagonalized-new}, by substituting $\alpha_j$ and $\widetilde{\alpha}_j$ as the eigenvalues of $(A,B)$ and $(\widetilde{A}, \widetilde{B})$ respectively.
This immediately gives the desired bound on $|\alpha_j - \widetilde{\alpha}_j|$.
\end{proof}

Lemma \ref{lem:exponent-perturbation} provides sufficient conditions for the computed exponents $\widetilde{\lambda}_j$ to be accurate.
It remains to analyze the resulting error in the coefficients $p_j$.
To do so, we recall the following result, attributed to Weyl:
\begin{theorem}[Singular Value Stability]
\label{thm:weyl-singular-value}
Let $A$ and $B$ be $n \times n$ matrices with entries in $\R$. Then for $j = 1, \ldots, n$, we have
$| \sigma_j(A+B) - \sigma_j(A) | < \|B\|$.
\end{theorem}

\begin{lemma} \label{lem:coeff-bound}
Consider step (3) of the Matrix Pencil Method, using $\widetilde{v}_0 \ldots, \widetilde{v}_{2k-1}$ and $\widetilde{\alpha}_1, \ldots, \widetilde{\alpha}_k$ in place of their true counterparts.
Let $\rho := \max_{i} |\alpha_i - \widetilde{\alpha}_i|$,
$\epsilon := \sup_{t \geq 0} |v_t - \widetilde{v}_t|$,
and $\widetilde{V} := V_{k}(\widetilde{\alpha}_1, \ldots, \widetilde{\alpha}_n)$.
Then if $\sigma_{min}(V) > k \rho$,
\begin{equation}
\label{eq:lemma-coeff-bound}
	\| \widetilde{p} - p \|_\infty \leq \frac{\sqrt{k}\epsilon + k \rho}{\sigma_{\min}(V) - k \rho}.
\end{equation}
\end{lemma}
\begin{proof}
Let $\eta = \widetilde{v} - v$, and recall that $v = Vp$, where $p$ is the vector $(p_1, \ldots, p_n)$.
By the setup, we may write $\widetilde{v} = Vp + \eta = \widetilde{V}p + (V - \widetilde{V})p + \eta$.
Let $\tilde{p}$ be such that $\widetilde{v} = \widetilde{V}\widetilde{p}$.
Equating the two formulae and moving terms around gives
\[
	\widetilde{p} - p
    =
    \widetilde{V}^{-1}\left( (V - \widetilde{V})p + \eta \right)
\]
which yields
\begin{equation}
\label{eq:lemma-coeff-bound-proof}
	\| \widetilde{p} - p \|
    \leq
    \|\widetilde{V}^{-1}\| (\|V - \widetilde{V}\| \|p\| + \|\eta\|)
    = \frac{ \|V - \widetilde{V}\| \|p\| + \|\eta\| }{\sigma_{\min}(\widetilde{V})}.
\end{equation}
Now observe that
\begin{enumerate}
	\item By Gershgorin's circle theorem, $\|V - \widetilde{V}\| \leq k\|V - \widetilde{V}\|_\infty = k \cdot \max_{i} |\alpha_i - \widetilde{\alpha}_i| = k\rho$.
    
    \item $p$ is a vector $(p_1, \ldots, p_n)$ such that $\sum_i p_i = 1$, so $\|p\| \leq 1$.
    
    \item $\|\eta\| \leq \sqrt{k}\|\eta\|_\infty \leq \sqrt{k} \epsilon$.
    
    \item Using Theorem \ref{thm:weyl-singular-value}, we get 
    $\sigma_{\min}(\widetilde{V}) \geq 
    \sigma_{\min}(V) - k\|\widetilde{V} - V\|_\infty$, given the hypothesis $\sigma_{\min}(V) > k\rho$.
\end{enumerate}
These, together with Equation~\ref{eq:lemma-coeff-bound-proof} and the bound $\|\widetilde{p} - p\|_\infty < \|\widetilde{p} - p\|$ give the desired result.
\end{proof}

The hypothesis $\sigma_{\min}(V) > k\rho$ of Lemma~\ref{lem:coeff-bound} is satisfied for sufficiently small $\beta$.
Qualitatively, this is the regime where $\beta$ is small enough so that $\sigma_{\min}(\widetilde{V}) \approx \sigma_{\min}(V)$.
Indeed, the statement of Theorem $\ref{thm:matrix-pencil-sample-complexity}$ specifies the case for large $N$, which will induce such a scenario.
To complete the analysis, we recall the Dvoretzky-Kiefer-Wolfowitz concentration inequality which gives uniform control of the error for the empirical CDF:
\begin{theorem}[DKW Inequality \cite{dvoretzky1956asymptotic,massart1990tight}]\label{thm:dkw}
Suppose $X_1, \ldots, X_n$ are i.i.d. samples from an unknown distribution
with CDF $F$. Let $F_n$ denote the $n$-sample empirical CDF, i.e.
\[ F_n(x) := \frac{1}{n} \#\{i : X_i \le x \}. \]
Then for every $\epsilon > 0$, 
\[
	\Pr \left( \sup_{t \in \mathbb{R}} | F(t) - F_n(t) | \ge \epsilon \right) \leq 2e^{-2n\epsilon^2}.
\]
\end{theorem}
We are now ready to finish the proof of Theorem~\ref{thm:matrix-pencil-sample-complexity}.
\begin{proof}[Proof of Theorem~\ref{thm:matrix-pencil-sample-complexity}]
Recall that $v_t$ is the tail probability $\Pr[ T \geq t ] = 1 - F(t)$, while the empirical estimate of $v_t$ is computed as $\widetilde{v}_t = 1 - F_n(t)$; therefore $|F(t) - F_n(t)| = |v_t - \widetilde{v}_t|$. 
Let $\epsilon = \sqrt{\frac{1}{2N} \log \frac{2}{\delta}}$, so that by the DKW inequality (\cref{thm:dkw}),
\[ \sup_t |v_t - \widetilde{v}_t| \le \epsilon \]
with probability at least $1 - \delta$.
Let $\beta = \frac{\epsilon k^2}{p_{\min}} \left(\frac{2e}{\Delta}\right)^{2k}$, as in the hypothesis of Lemma~\ref{lem:exponent-perturbation}.
Observe that both $\epsilon$ and $\beta$ decrease to zero as $N$ increases for any fixed $\delta$.
Specifically, we may take $N$ large enough so that $\frac{1}{\beta} - 1 > \frac{1}{2\beta}$.

According to Lemma~\ref{lem:exponent-perturbation}, the event $| v_t - \widetilde{v}_t | < \epsilon$ implies the desired bound on the $\alpha$'s:
\begin{align*}
	| \alpha_j - \widetilde{\alpha}_j |
    \leq k^{1.5} \frac{ \frac{1}{p_{\min}} + 1 }{ \frac{1}{\beta} - 1 }
    \leq k^{1.5} \cdot \frac{2}{p_{\min}} \cdot 2\beta
    = \frac{4 k^{3.5} \epsilon}{p_{\min}^2} \left( \frac{2e}{\Delta} \right)^{2k}.
\end{align*}

To get the bound on the $p$'s, again we require $N$ sufficiently large, this time so that 
$\frac{1}{\epsilon\sqrt{k}\left(\frac{2e}{\Delta}\right)^{k}} - \frac{4k^{4.5}}{p_{\min}^2}\left( \frac{2e}{\Delta} \right)^{2k} \geq \frac{1}{2\epsilon\sqrt{k}\left(\frac{2e}{\Delta}\right)^{k}}$.
Such a condition ensures that $\frac{4k^{4.5} \epsilon}{p_{\min}^2}\left( \frac{2e}{\Delta} \right)^{2k} < \sigma_{\min}(V)$ by Lemma~\ref{lem:singular-value-vandermonde}, which allows us to apply Lemma~\ref{lem:coeff-bound}.
This gives us
\[
	| p_j - \widetilde{p}_j |
    < \frac{\sqrt{k} \epsilon + \frac{4 k^{4.5} \epsilon}{p_{\min}^2} \left( \frac{2e}{\Delta} \right)^{2k}}{\sigma_{\min}(V) - \frac{4 k^{4.5} \epsilon}{p_{\min}^2} \left( \frac{2e}{\Delta} \right)^{2k}}
    \leq \frac{\frac{5 k^{4.5}}{p_{\min}^2} \left( \frac{2e}{\Delta} \right)^{2k}}{\frac{1}{\epsilon\sqrt{k}} \left( \frac{2e}{\Delta} \right)^{-k} - \frac{4 k^{4.5}}{p_{\min}^2} \left( \frac{2e}{\Delta} \right)^{2k}} 
    \leq \frac{10k^{5} \epsilon}{p_{\min}^2} \left( \frac{2e}{\Delta} \right)^{3k}.
\]
In particular, both of the stated conditions on $N$ are satisfied for $N \geq N_0$, where
\[ 
	N_0 
    = 
    \frac{ 32k^{10} }{ p_{\min}^4 } \left( \frac{2e}{\Delta} \right)^{6k} \log \frac{2}{\delta}.
\]

\end{proof}


\input{info-lower-bound-revised}

\section{Hypothesis Testing using the Nazarov-T\'uran Lemma}
\label{sec:nazarov-turan}
In this section we use the Nazarov-T\'uran lemma, originating in works
on analytic number theory and approximation theory, to give
guarantees on the number of samples needed for the Kolmogorov-Smirnov
test described in \cref{sec:results-nt} to successfully solve the simple-vs-composite hypothesis testing Problem~\ref{problem:testing2}.

As a first step towards proving \cref{thm:nt-hypothesis-testing}, we now recall our key technical tool, the Nazarov-Tur\'an Lemma\footnote{Here we have specialized the Nazarov-Tur\'an lemma to real exponents and to intervals instead of general measurable sets.}.
\begin{theorem}[Nazarov-Tur\'an \cite{nazarov1993local}]\label{thm:nazarov-turan}
Suppose $p(t) = \sum_{j = 1}^n c_j e^{-\lambda_j t}$. Then for $t_1,t_2 > 0$
\[ \max_{t \in [-t_2,t_1]}|p(t)| \le e^{(t_1 + t_2)\lambda_1} [4e (t_1 + t_2)/t_1]^{n - 1} \max_{t \in [0, t_1]} |p(t)|. \]
\end{theorem}
We will use the Nazarov-Tur\'an lemma by analytically continuing the mixture of exponential CDFs to negative time; this allows us to derive the following lemma
which is the heart of our hypothesis testing result.
\begin{lemma}\label{lem:nazarov-turan-lower-bound}
Suppose $p(t) = \sum_{j = 1}^{k_0} c_j e^{-\lambda_j t} - \sum_{j = 1}^{k_1} d_j e^{-\mu_j t}$ such that the $c_j, d_j \ge 0$ and $\sum_j c_j, \sum_j d_j \le 1$. Also suppose that $1 \ge \lambda_1 > \cdots > \lambda_{k_0} \ge 0$ and similarly for the $\mu_j$. Then if $k := \frac{k_0 + k_1}{2}$ and $\delta := \min(|c_1|,|d_1|)$
\[ \frac{\delta}{2(8e^2 \log(2/\delta))^{2k-1}} \cdot \min\left(1, ((2k-1)(\lambda_1 - \mu_1))^{2k-1}\right) \le \max_{t \in [0,2k - 1]} |p(t)|
\]
\end{lemma}
\begin{proof}
Without loss of generality we suppose that $\lambda_1 > \mu_1$. 
By \cref{thm:nazarov-turan}, for any $t_1,t_2$ we have
\[ 
	|p(-t_2)| 
    \le e^{(t_1 + t_2) \lambda_1} [4e (t_1 + t_2)/t_1]^{2k - 1} \max_{t \in [0, t_1]} |p(t)|.
\]
To use this, we will need to choose a $t_2$ so that we can get an explicit lower bound on $p(-t_2)$. It will suffice if
\[ 
|c_1 e^{\lambda_1 t_2}|/2 \ge |\sum_{j = 1}^{k_1} d_j e^{\mu_j t_2}| \]
or, phrased slightly differently:
\begin{equation}
\label{eqn:leading-dominates}
	|c_1|/2 
    \ge 
    |\sum_{j = 1}^{k_1} d_j e^{(\mu_j - \lambda_1) t_2}|
\end{equation}
Observe that all the exponents on the right hand side have negative coefficients, because $\lambda_1$ is the largest.
In order to bound the right hand side, we use Holder's Inequality:
\[ 
|\sum_{j = 1}^n d_j e^{(\mu_j - \lambda_1) t_2}| 
\le e^{(\mu_1 - \lambda_1)t_2} \sum_{j = 1}^{k_1} |d_j| \le e^{(\mu_1 - \lambda_1)t_2}.
\]
Therefore, taking $t_2 = \frac{\log 2/|c_1|}{\lambda_1 - \mu_1}$ suffices for \eqref{eqn:leading-dominates} to hold.
We then have
\begin{align*}
	& |c_1 e^{\lambda_1 t_2}|/2 \le e^{(t_1 + t_2) \lambda_1} [4e (t_1 + t_2)/t_1]^{2k - 1} \max_{t \in [0, t_1]} |p(t)| \\
    \implies
    & |c_1|/2 \le e^{t_1 \lambda_1} [4e (t_1 + \frac{\log(2/|c_1|)}{\lambda_1 - \mu_1})/t_1]^{2k - 1} \max_{t \in [0, t_1]} |p(t)|.
\end{align*}
which gives a lower bound on the maximum value of $|p(t)|$:
\[ \frac{|c_1|}{2} e^{-t_1 \lambda_1} \left[\frac{t_1}{4e (t_1 + \frac{\log(2/|c_1|)}{\lambda_1 - \mu_1})}\right]^{2k - 1} \le \max_{t \in [0,t_1]} |p(t)|. \]
Letting $t_1 = (2k - 1)$ and using that $|\lambda_1| \le 1$, we find
\begin{equation}\label{eqn:simple-lb}
\frac{|c_1|}{2(4e^2)^{2k - 1}} \left[\frac{2k - 1}{(2k - 1)+ \frac{\log(2/|c_1|)}{\lambda_1 - \mu_1}}\right]^{2k - 1} \le \max_{t \in [0,2k - 1]} |p(t)|
\end{equation}
Using the inequality that $x+y \leq 2 \max(x,y)$, we bound the second term on the LHS as follows:
\begin{align}
\left[\frac{2k - 1}{(2k - 1)+ \frac{\log(2/|c_1|)}{\lambda_1 - \mu_1}}\right]^{2k - 1} &\ge \frac{1}{2^{2k-1}}\left[\frac{2k - 1}{\max\left(2k-1,  \frac{\log(2/|c_1|)}{\lambda_1 - \mu_1}\right)}\right]^{2k - 1} \nonumber \\
&\ge \frac{1}{2^{2k-1}}\min\left(1, \frac{(2k-1)(\lambda_1 - \mu_1)}{\log(2/|c_1|)}\right)^{2k-1} \nonumber \\
&\ge \frac{1}{(2\log(2/|c_1|)^{2k-1}} \min\left(1, ((2k-1)(\lambda_1 - \mu_1))^{2k-1}\right)\label{eqn:lb-frac}
\end{align}

Using Equation~\ref{eqn:lb-frac}, we simplify the lhs of \eqref{eqn:simple-lb} to give
\[ \frac{|c_1|}{2(8e^2 \log(2/|c_1|))^{2k-1}} \cdot \min\left(1, ((2k-1)(\lambda_1 - \mu_1))^{2k-1}\right) \le \max_{t \in [0,2k - 1]} |p(t)|. \]
\end{proof}
We are now ready to prove the finite sample guarantee (\cref{thm:nt-hypothesis-testing}) for the Kolmogorov-Smirnov test from \cref{sec:results-nt}.
For the reader's convenience, we re-state the test now:
\begin{center}
\fbox{\parbox{\textwidth}{
\begin{enumerate}
\item Let $\alpha > 0$ be the significance level.
\item Let $F_N$ be the empirical CDF and let $F$ be the CDF under the null hypothesis $H_0$.
\item Reject $H_0$ if $\sup_t |F_n(t) - F(t)| > \sqrt{\log(2/\alpha)/2N}$.
\end{enumerate}
}}
\end{center}

\begin{proof}[Proof of \cref{thm:nt-hypothesis-testing}]
Let $F$ be the CDF under the null hypothesis, that the samples are drawn from $P$.
By the DKW inequality (\cref{thm:dkw}), under the null hypothesis
\[ \Pr_{P}(\sup_t |F_N(t) - F(t)| > s) \le 2e^{-2N s^2} \]
Therefore with $s = \sqrt{\log(2/\alpha)/2N}$ then the probability of type I error is bounded by $\alpha$, as desired. On the other hand, observe by Lemma~\ref{lem:nazarov-turan-lower-bound} that if $F_Q$ is the CDF of $Q \in H_1$, then
\[ \frac{\delta}{2(8e^2)^{2k-1}} \cdot \min\left(1, ((2k-1)(\epsilon))^{2k-1}\right) \le \max_{t \in [0,2k - 1]} |F(t) - F_Q(t)| \]
It follows that the test is guaranteed to reject as long as
\begin{equation}  
\label{eq:guarantee-1}
\frac{\delta}{2(8e^2)^{2k-1}} \cdot \min\left(1, ((2k-1)(\epsilon))^{2k-1}\right) - \sup_t |F_N(t) - F_Q(t)| >  \sqrt{\log(2/\alpha)/2N}. 
\end{equation}
Under the alternative hypothesis, by the DKW Inequality (Theorem~\ref{thm:dkw}). Equation~\ref{eq:guarantee-1} happens with probability
\begin{align*}
	&\Pr\left(\sup_t |F_N(t) - F_Q(t)| <\frac{\delta}{2(8e^2)^{2k-1}} \cdot \min(1,((2k-1)\epsilon)^{2k-1}) - \sqrt{\log(2/\alpha)/2N}\right) \\
    \ge &1 - 2\exp\left(- \frac{N\delta^2 \min(1, ((2k-1)\epsilon)^{4k-2})}{8(8e^2)^{4k-2}}\right)
\end{align*}
since $\sqrt{\log(2/\alpha)/2N} < \frac{\delta \cdot \min(1,((2k-1)\epsilon)^{2k-1})}{4(8e^2)^{2k-1}}$ under the assumption $N \ge N_0$.
\end{proof}
\begin{remark}
The point of the above argument is to derive finite sample guarantees.
In the asymptotic regime, one can immediately derive more precise results by combining the guarantee from \cref{lem:nazarov-turan-lower-bound}, applied to the difference of CDFs, with the classical version of the Kolmogorov-Smirnov test, using the critical values for the Kolmogorov distribution. 
\end{remark}



\input{population-algorithm}
\section{Simulation Methods}
\subsection{A convex programming approach to learning}
\label{sec:convex-program-desc}
In addition to using the Matrix Pencil Method to learn mixtures of exponentials in each interval, we also implemented a convex program.
Here, the goal is to learn a mixture of exponentials, whose support is perhaps restricted to an interval $I = [a,b]$.
The idea is as follows: assume that we know the interval $\Lambda = [0,c]$ for which we can assume $\lambda_1, \ldots, \lambda_n \in \Lambda$.
We first discretize the space of possible exponents by choosing $n$ equally spaced points $\lambda_1, \ldots, \lambda_n$ inside $\Lambda$.
Solve the convex program
\begin{equation*}
\begin{aligned}
& \underset{\vec{p}}{\text{minimize}}
& & \sup_{t \in I}\left| \sum_{i=1}^{n} p_i e^{\lambda_i t} - v_t \right| \\
& \text{subject to}  & & \sum_i p_i = 1 \\
& & & p_i \ge 0, i = 1,\ldots,n. \\
& & 
\end{aligned}
\end{equation*}
In practice, we replace $\sup_{t \in I}$ with the discretization $\max_{t \in \mathcal{S}}$, where $\mathcal{S} \subset I$ is a finite mesh of points in $I$.
Since we are learning from samples, we also substitute $v_t$ with $\widetilde{v}_t$, the empirical estimate of the tail probability $\Pr[T=t \mid T \geq a]$. Since the $\ell_1$-norm of the $p_i$ is fixed to be $1$, we do not expect to need additional regularization to get sparse output.

For small instances (see \cref{sec:reconstruction-sims}), the convex program is more sample-efficient than the Matrix Pencil Method.
In the context of this paper, however, it does not come with robustness guarantees.
Results on convex programming approaches for super-resolution are known, due to Cand{\`e}s and Fernandez-Granda \cite{candes2013super}; for our (real-exponent) setting, a different analysis will be required and we leave this to future work. If we assume that the program does return a sparse output (which occurs in practice), some guarantees for the accuracy of the output follow automatically from the analysis of Theorem~\ref{thm:matrix-pencil-sample-complexity} and Theorem~\ref{thm:nt-hypothesis-testing}, since for sparse mixtures they (implicitly) bound parameter error in terms of the closeness in CDF-distance.

\textbf{Implementation in simulations: } In our experiments, we solved the above convex program using the barrier method of CPLEX version 12.8 with numerical emphasis enabled.
\subsection{Simulations: Additional Details}\label{sec:simulation-math}
\textbf{Earthmover's distance between parameters: }
The Earthmover's (or 1-Wasserstein) distance between $P$ and $Q$ measures the minimum transport cost to move the ``mass'' corresponding to probability distribution $P$ to that of $Q$. Rigorously, in one dimension it can be defined by
\[ EMD(P,Q) := \min_{\pi : \pi|_X = P, \pi|_Y = Q} \E_{(X,Y) \sim \pi}[|X - Y|]\]
where here $\pi$ ranges over all possible couplings of marginal distributions $P$ and $Q$.
The following definition makes the notion of Earthmover's distance between the parameters of two mixtures of exponentials precise:
\begin{defn}
Let $P$ and $Q$ be two mixtures of exponentials $P(T > t) = \sum_i p_i e^{-\lambda_i t}$ and $Q(T > t) = \sum_i q_i e^{-\gamma_i t}$. The \emph{Earthmover's distance in parameter space} between $P$ and $Q$ is the Earthmover's distance between corresponding atomic measures $\mu_P := \sum_i p_i \delta_{\lambda_i}$ and $\mu_Q := \sum_i q_i  \delta_{\gamma_i}$ where $\delta_x$ represents a Dirac mass at point $x$. 
\end{defn}

\noindent
\textbf{Derivation of Per-Instance Information-Theoretic Lower bounds: } Given the alternative instance, we derived the bound by computing the $H^2$ (Hellinger squared) distance between the true distribution and the alternative distribution, and then applying standard tensorization and comparison inequalities to bound the TV (as used in the proof of Corollary~\ref{cor:turan-tv}). 
\\\\
\noindent
\textbf{Upper bound simulations: } We ran 300 trials for each setting of $k$ and number of samples; in order to run the simulation for very large numbers of samples, we directly generated the corresponding noisy CDF estimates by adding Gaussian noise of order $O(1/\sqrt{N})$ where $N$ is the number of samples. For reasonable size $N$ we also ran the methods using actual sample-estimated CDFs and the results were consistent with the simulated Gaussian-noise CDFs. The lower bound is analytically computed, not simulated, so it is unaffected by this Gaussian-noise approximation.
\\\\
\noindent
\textbf{Plotted Data:} Here we provide the data plotted in Figure~\ref{fig:components-vs-logsamples}, that was found via simulation as described above.
\begin{table}[h]
    \centering
    \begin{tabular}{|c|c|c|c|}
        \hline
         k & CVX & MPM & LB \\
         \hline
         1 & $2.98 \times 10^{5}$ & $9.28 \times 10^4$ & $1.34 \times 10^4$ \\
         \hline
         2 & $3.25 \times 10^{8}$ & $3.45 \times 10^{10}$ & $8.18 \times 10^6$ \\
         \hline
         3 & $3.55 \times 10^{11}$ & $3.87 \times 10^{14}$ & $1.44 \times 10^{8}$ \\
         \hline
         4 & $1.21 \times 10^{14}$ & $1.40 \times 10^{19}$ & $1.13 \times 10^{9}$ \\
         \hline
         5 & N/A & $4.89 \times 10^{22}$ & $1.43 \times 10^{13}$\\
         \hline
    \end{tabular}
    \caption{Values plotted on a log (base 10) scale in Figure~\ref{fig:components-vs-logsamples}.}
    \label{table:appendix}
\end{table}

\end{document}

%% file: intro-new-2.tex
\subsection{Background: inference of population size history}
\label{sec:background1}


A central task in population genetics is to reconstruct a species' {\em effective population size} over time.
Coalescent theory \cite{nordborg2001coalescent} provides a mathematical framework for understanding the relationship between effective population size and genetic variability. 
In this framework, observations of present-day genetic variability \---- captured by DNA sequences of individuals \---- can be used to make inferences about changes in population size over time.




There are many existing methods for estimating the size history of a \emph{single population} from \emph{sequence data}. Some rely on Maximum Likelihood methods \cite{li2011inference, schiffels2014inferring, Sheehan2013, terhorst2017robust} and others utilize Bayesian inference \cite{Drummond2005, heled2008bayesian, Nielsen2000} along with a variety of simplifying assumptions.
A well-known work of Li and Durbin \cite{li2011inference} is based on using sequence data from just a single human (a single pair of haplotypes) and revolves around the assumption that coalescent trees of alleles along the genome satisfy a certain conditional independence property \cite{mcvean2005approximating}. 
By and large, methods such as these do not have any associated provable guarantees. For example, Expectation-Maximization (EM) is a popular heuristic for maximizing the likelihood but can get stuck in a local maximum. Similarly, Markov Chain Monte Carlo (MCMC) methods are able to sample from complex posterior distributions if they are run for a long enough time, but it is rare to have reasonable bounds on the mixing time. In the absence of provable guarantees, simulations are often used to give some sort of evidence of correctness.




Under what sorts of conditions is it possible to infer a single population history? Kim, Mossel, R{\'a}cz and Ross \cite{kim2015can} gave a strong lower bound on the number of samples needed even when one is given exact coalesence data. In particular, they showed that the number of samples must be at least exponential in the number of generations. Thus there are serious limitations to what kind of information we can hope to glean from (say) sequence data from a single human individual. In a sense, their work provides a quantitative answer to the question: {\em How far back into the past can we hope to reliably infer population size, using the data we currently have?} We emphasize that although they work in a highly idealized setting, this only makes their problem easier (e.g. assuming independent inheritance of loci along the genome and assuming that there are no phasing errors) and thus their lower bounds more worrisome.



\subsection{Our setting: inference of multiple subpopulation histories}

A more interesting and challenging task is the reconstruction of \textit{population structure}, which refers to the sub-division of a single population into several subpopulations that merge, split, and change sizes over time. There are two well-known works that attack this problem using coalescent-based approaches.
Both use sequence data to infer population histories where present-day subpopulations were formed via divergence events of a single, ancestral population in the distant past.
The first is Schiffels and Durbin \cite{schiffels2014inferring}, who used their method to infer the population structure of nine human subpopulations up to about 200,000 years into the past.
More recently, Terhorst, Kamm and Song \cite{terhorst2017robust} inferred population structures of up to three human subpopulations.
Just as in the single population case, these methods do not come with provable guarantees of correctness due to the simplifying assumptions they invoke and the heuristics they employ.

As for theoretical work, the lower bounds proven for a single population trivially carry over to the setting of inferring population structure. However, the lower bound in~\cite{kim2015can} only applies when we are trying to reconstruct events in the distant past, leading us to a natural question: can we infer \emph{recent} population structure, but, when there are multiple subpopulations?

In this paper, we establish strong limitations to inferring the population sizes of multiple subpopulation histories using pairwise coalescent trees.
We prove sample complexity lower bounds that are \emph{exponential in the number of subpopulations}, even for reconstructing recent histories. 
Our results provide a quantitative answer to the question, 
{\em Up to what granularity of dividing a population into multiple subpopulations, can we hope to reliably infer population structure?} 

Our methods incorporate tools from information theory,
approximation theory, and analysis (from \cite{turan1984new}).
To complement our lower bounds, we also give an algorithm for hypothesis testing based on the celebrated Nazarov-Tur\'an lemma~\cite{nazarov1993local}. Our upper and lower bounds match up to constant factors and establish sharp bounds for the number of samples needed to distinguish between two known population structures as a function of the number of subpopulations. 
Finally, for the more general problem of learning the population structure (as opposed to testing which of two given population structures is more accurate) we give an algorithm with provable guarantees based on the Matrix Pencil Method~\cite{MPM} from signal processing.
We elaborate on our results in Section~\ref{sec:results}.

%% file: coalescent.tex
\subsubsection{The Multiple-Subpopulation Coalescent Model}
\label{sec:coalescent}
Consider a panmictic haploid\footnote{In a diploid population, the exponents are scaled by a constant factor $2$. This can be handled easily via scaling and therefore makes little difference in the analysis.} population, such that each subpopulation evolves according to the standard Wright-Fisher dynamics\footnote{The distinction between the Wright-Fisher and Moran models is of no consequence in this work, as the latter yields the same exponential model in the diffusion limit \cite{blythe2007stochastic}.} \---- we direct the reader to \cite{blythe2007stochastic} for an overview.
For simplicity, we assume no admixture between distinct subpopulations as long as they are separated in the model (i.e. they have not merged into a single population in the time period under consideration).

As a reminder, if one assumes that a single population has size $N$ which is large and constant throughout time, then the time to the most recent common ancestor (TMRCA) of two randomly sampled individuals closely follows the Kingman coalescent \cite{blythe2007stochastic} with exponential rate $N$:
\begin{equation}
\label{eq:coalescent-regular}
	\Pr( T > t ) = \exp(- t/N).
\end{equation}
where $T$, the coalescence time for two randomly chosen individuals, is measured in \textbf{generations}.
Henceforth, we will assume that this is the distribution of $T$ in the single-component case.

If instead we have a population which is partitioned into a collection of distinct subpopulations with non-constant sizes, let $\vN(t)$ be the function that describes the sub-population sizes over time.
As in \cite{kim2015can}, we will assume that the function $\vN(t)$ is piecewise constant with respect to some unknown collection of intervals $I_1, I_2, \ldots$ partitioning the real line.
In particular, for each $t \in I_k$, there is an associated vector of effective subpopulation sizes $\vN(t) = (N_1^{(k)}, \ldots, N_{D_k}^{(k)})$, indexed by the $D_k$ subpopulations present at time $t$.
The indexing need not be consistent across different intervals, as their semantic meaning will change as subpopulations merge and split.
For example, $N_1^{(k)}$ and $N_1^{(k+1)}$ need not always represent the sizes of the same subpopulation.

Consider the case where $\vN(t)$ is constant for all $t \in I = [a,b]$, where $0 < a < b$, with no admixture and no migration in-between subpopulations in the time interval $I$.
In this case, the coalescence time follows the law of a convex combination of exponential functions:
\begin{equation} \label{eq:coalescent-mixture}
	\Pr(T > a + t \;\big|\; T > a)
	=
	\sum_{\ell = 0}^{D} p_{\ell} e^{-\lambda_\ell t}
\end{equation}
where $p_0 + p_1 + \cdots + p_D = 1$, $\lambda_0 = 0$ and the other $\lambda_i$ are $\frac{1}{N_i}$ (refer to \cref{sec:model-derivation} for a more careful treatment). 

The population structure is assumed to undergo changes over time, where the positive direction points towards the past. 
The three possible changes are:
\begin{enumerate}
	\item (Split) One subpopulation at time $t^-$ becomes two subpopulations at time $t$ (i.e. $D_{k} = D_{k-1} + 1$).
	\item (Merge) Two subpopulations at time $t^-$ join to form one subpopulation at time $t$ (i.e. $D_{k} = D_{k-1} - 1$).
	\item (Change Size) An arbitrary number of subpopulations change size at time $t$. 
\end{enumerate}
Figure~\ref{fig:pophist} provides an illustrative example.
If an individual at time $t^{-}$ is from a subpopulation of size $M$ which splits into two subpopulations of sizes $M_1, M_2$ at time $t$, then its ancestral subpopulation is random: for $i \in \{1,2\}$, subpopulation $i$ is chosen with probability $M_i / M$.
In our model, we only allow at most one of these events at any particular time point.
For us, a ``split'' looking backward in time refers to a \emph{convergence} event of two subpopulations going \emph{forward} in time, while a ``merge'' refers to a divergence event.
This convention is chosen because we think of reconstruction as proceeding backwards in time from the present.

\begin{figure}
\vspace{-1mm}
\centering\includegraphics[scale=0.3]{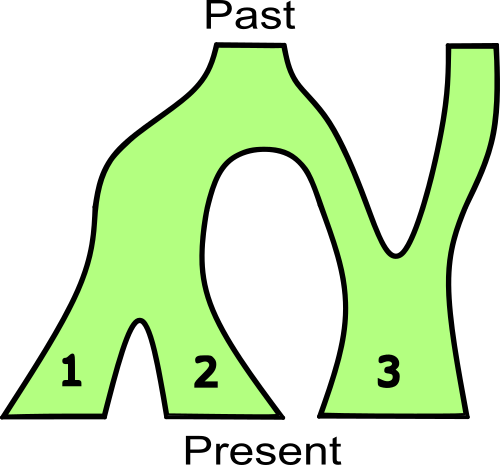}
\caption{An example of population structure history, illustrating merges and splits starting with three present-day subpopulations.}
\label{fig:pophist}
\end{figure}

%% file: related-works.tex
\subsection{Related works}

As mentioned in \cref{sec:background1}, existing methods that attempt to empirically estimate the population history of a \emph{single population} from \emph{sequence data} generally fall into one of two categories: Many are based on (approximately) maximizing the likelihood \cite{li2011inference, schiffels2014inferring, Sheehan2013, terhorst2017robust} and others perform Bayesian inference \cite{Drummond2005, heled2008bayesian, Nielsen2000}. 
Generally, they are designed to recover a \textit{piecewise constant} function $N(t)$ that describes the size of a population, with the goal of accurately summarizing divergence events, bottleneck events and growth rates throughout time.

Many notable methods that fall into the first category rely on Hidden Markov Models (HMMs), which implicitly make a Markovian assumption on the coalescent trees of alleles across the genome.
One notable work is Li and Durbin \cite{li2011inference}, which gave an HMM-based method (PSMC) that reconstructs the population history of a single population using the genome of a single diploid individual.
Later related works gave alternative HMMs that incorporate more than two haplotypes (diCal \cite{Sheehan2013} and MSMC \cite{schiffels2014inferring}) and improve robustness under phasing errors (SMC++ \cite{terhorst2017robust}).

Methods in the second category operate under an assumption about the probability distribution of coalescence events and the observed data.
For instance, Drummond \cite{Drummond2005} prescribes a prior for the distribution of coalescence trees and population sizes, under which MCMC techniques are used to compute both an output and a corresponding 95\% credibility interval.
However, given the highly idealized nature of their models and the limitations of their methodology (for example, there is no guarantee their MCMC method has actually mixed), it is unclear whether the ground truth actually lies in those credibility intervals.

In the \emph{multiple subpopulations} case, there are two major coalescent-based methods.
The first is Schiffels and Durbin \cite{schiffels2014inferring}, which introduced the MSMC model as an improvement over PSMC.
These authors used their method to infer the population history of \emph{nine human subpopulations up to about 200,000 years into the past}.
Terhorst, Kamm and Song \cite{terhorst2017robust} introduced a variant (SMC++) that was directly designed to work on genotypes with missing phase information.
In particular, they demonstrate the potential dangers of relying on phase information, by showing that MSMC is sensitive to such errors.
In an experiment, SMC++ was used to perform inference of population histories of various combinations of \emph{up to three human subpopulations}.
In these experiments, individuals are purposefully chosen from specific subpopulations. We emphasize that in our model, due to the presence of population merges \emph{and} splits, one does not always know what subpopulation an ancestral individual is from.



As a side remark, there are approaches that attempt to infer a (single-component) population history using different types of information.
We briefly touch upon some of these known works.
One alternative strategy is to use the site frequency spectrum (SFS), e.g. \cite{bhaskar2015efficient, Excoffier2013}.
The earliest theoretical result regarding SFS-based reconstruction is due to Myers, Fefferman and Patterson \cite{myers2008can}, who proved that \emph{generic} 1-component population histories suffer from unidentifiability issues.
Their lower bound constructions have a caveat: They are pathological examples of oscillating functions which are unlikely to be observed in a biological context.
Later works \cite{bhaskar2014descartes, terhorst2015fundamental} prove both identifiability and lower bounds for reconstructing \emph{piecewise constant} population histories using information from the SFS. (In contrast, as our algorithms show, reconstruction from coalescence data does not suffer from the same lack of identifiability issues.)

Most recently, Joseph and Pe'er \cite{joseph2018inference} developed a Bayesian time-series model that incorporates data from \emph{ancient DNA} to recover the history for multiple subpopulations only under size changes, without considering merges or splits.
While our analysis does not directly account for such data, the necessity of considering such models is consistent with our assertion: extra information about the ground truth, such as directly observable information about the past (e.g. ancestral DNA), is probably \emph{required} in order for the problem to even be information-theoretically feasible.
In addition, \cite{joseph2018inference} does not solve for subpopulation \emph{sizes}, but rather subpopulation \emph{proportions}, which contains less information than what we are after.

%% file: simulations.tex
\section{Simulations and Indistinguishability in Simple Examples}
\label{sec:reconstruction-sims}
Our theoretical analysis rigorously establishes the worst-case
dependence on the number of samples needed in order to learn the parameters
of a single period of population history under our model -- recall the construction
of \cref{thm:turan-sample-complexity-intro} of two hard-to-distinguish mixtures of exponentials
and the result \cref{thm:pop-reduction} converting these to population histories.

In our simulations, we will analyze both the performance and information-theoretic difficulty of learning not a specially constructed worst-case instance, but instead an \emph{extremely simple population history} with $k$ populations.  More precisely we consider the following instance:

\vspace{.20cm}
\noindent
\textbf{Simulation Instance($k$)}:
\begin{itemize}
\item[] \textbf{Population history description: } We consider reconstructing a single period model with $k$ populations in which the ratio of the population sizes is $1:2:\cdots:k$ and the relative probability of tracing back to each of these populations (i.e. $\Pr(\mathcal{E}_{i,i} | T > t_0)$ from \cref{sec:model-derivation}) are all equal to $1/k^2$. This can easily be realized as a one period of a 2-period population history model, in which in the second (more recent) era all populations are the same size\footnote{As in Remark~\ref{rmk:short-modern-era}, we can optionally make the more recent era short so that almost all samples will be from the earlier period.} 
\item[] \textbf{Mixture of exponentials description: } We consider the following mixture of exponentials:
\[ \Pr(T > t) = (1 - 1/k) + \sum_{i = 1}^k (1/k^2) e^{-t/k}. \]
The constant term represents atomic mass at $\infty$ and corresponds to no coalescence. When $k = 1$ this is a standard exponential distribution, otherwise it is a mixture of $k + 1$ exponentials, counting the degenerate constant term.  
\end{itemize}

We do not believe that this is an unusually difficult instance of a mixture of exponentials on $k$ components.
If anything, the situation is likely the opposite:
our worst-case analysis (Theorems~\ref{thm:minimax-lower-bound-intro}, ~\ref{thm:matrix-pencil-sample-complexity}) suggests that this is comparatively \emph{easy} as the gap parameter $\Delta$ is maximally large. 

In order to evaluate the error in parameter space from the result of the learning algorithm, we adopted a natural metric, the well-known \emph{Earthmover's distance}. 
Informally, this measures the minimum distance (weighted by $p_i$ and recovered $\tilde{p}_i$) that the recovered exponents must be moved to agree with the ground truth; we give the precise definition in Appendix~\ref{sec:simulation-math}. 

For a point of comparison to MPM, we also tested a natural convex programming formulation which essentially minimizes $\|\int e^{-\lambda t} d\mu(\lambda) - (1 - \tilde{F}(t))\|_{\infty}$ over probability measures $\mu$ on $\mathbb{R}_{\ge 0}$, where $\tilde{F}$ is the empirical CDF -- refer to \cref{sec:convex-program-desc} for details. 
The results of running both the convex program 
and the MPM are shown in Figure~\ref{fig:components-vs-logsamples} (blue and green lines) plotted on a log (base 10) scale; details of the setup are provided in Appendix~\ref{sec:simulation-math}. 
As expected based on our theoretical analysis, the number of samples needed scaled exponentially in $k$, the number of populations in our instance. 
Details of the setup are provided in Appendix~\ref{sec:simulation-math}; due to limitations of machine precision, the convex program could not reliably reconstruct at $5$ components with any noise level and so this point is omitted. 

Besides showing the performance of the algorithms, we were able to deduce \emph{rigorous, unconditional} lower bounds on the information-theoretic difficulty of these particular instances. Each point on the red line corresponds to the existence of a different mixture of exponentials (found by examining the output of the convex program), with a comparable number of mixture components\footnote{The alternative hypothesis had no more than a few additional mixture components.
A byproduct of this analysis is that even estimating the number of populations is in these examples requires a very large number of samples.}, which is far in parameter space\footnote{More precisely, with Earthmover's distance in parameter space greater than $0.01$. For comparison, an estimator which only gets the (easy) constant component correct already has Earthmover distance at most $1/k$.} from the ground truth and yet the distribution of $N$ samples from this model (where $N = 10^y$ and $y$ is the $y$-coordinate in the plot) has total-variation (TV) distance at most $0.5$ from the distribution of $N$ samples from the true distribution. 
By the Neyman-Pearson Lemma, this implies that if the prior distribution is $\left( \frac{1}{2}, \frac{1}{2} \right)$ between these two distributions, then we cannot successfully distinguish them with greater than $75\%$ probability. We describe the mathematical derivation of the TV bound in Appendix~\ref{sec:simulation-math}, and illustrate such a hard-to-distinguish pair in Example~\ref{example:hard-mixture}. Recall that by Theorem~\ref{thm:pop-reduction}, such a hard to distinguish pair of mixtures can automatically be converted into a pair of hard-to-distinguish population histories.

Notably, the lower bound shows that reliably learning the underlying parameters 
in this simple model with $5$ components necessarily requires at least 10 trillion samples from the true coalescence distribution. In reality, since we do not truly have access to clean i.i.d. samples from the distribution, this is likely a significant underestimate. 

\begin{figure}\label{fig:components-vs-logsamples}
\centering
\vspace{-10mm}
\includegraphics[width=0.65\textwidth]{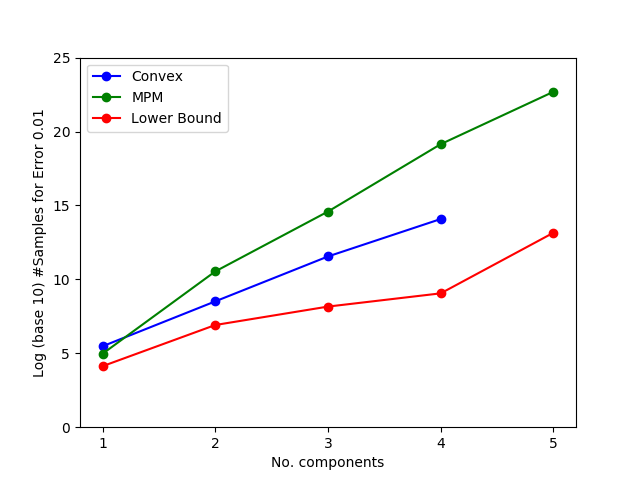}
\caption{Plot of \#components versus log (base 10) number of samples needed for accurate reconstruction (parameters within Earthmover's distance $0.01$). Below the red line, it is mathematically impossible for \emph{any method} to distinguish with greater than 75\% success between the ground truth and a fixed alternative instance which has significantly different parameters.}
\end{figure}

\begin{example}\label{example:hard-mixture}
Consider the mixtures of exponentials with CDFs $F(t)$ and $G(t)$, where
$1 - F(t) = 0.5 + 0.25e^{-0.5t} + 0.25e^{-t}$ and 
\[1 - G(t) = 0.49975946 + 0.15359557e^{-0.45t} + 0.30642727 e^{-0.81t} + 0.0402177e^{-1.55t}. \]
Despite being very different in parameter space, their $H^2$ distance is $7.9727 \cdot 10^{-6}$ so
any learning algorithm requires at least $15660$ 
samples to distinguish them with better than 75\% success rate. 

As a remark, we point out that the CDFs $F$ and $G$ in this example have exponents that are interlaced. 
Observe that this is a characteristic also shared by the information-theoretic obstructions referenced in \cref{sec:info-overview} and \cref{sec:info-lower-bound}.
This likely illustrates a major source of difficulty of most reasonable-looking instances: ``averaging" adjacent exponents of an exponential mixture may produce a different mixture with a similar distribution whose components interlace with the original.
\end{example}

%% file: population-reduction.tex
\subsection{Hardness for Distinguishing Population Histories}

Theorem~\ref{thm:equivalence} shows that any arbitrary mixture of exponentials can be embedded as a sub-mixture of a simple population history with two time periods. We can leverage this equivalence to reduce distinguishing two mixtures of exponentials to distinguishing two population histories, and hence conclude from Theorem~\ref{thm:turan-sample-complexity-intro} that distinguishing two population histories is exponentially hard in the number of subpopulations.

The high-level idea is to take the reduction from Theorem~\ref{thm:equivalence} and apply it to two arbitrary mixtures of exponentials, and argue that the problem of distinguishing the resulting two population histories is at least as hard as distinguishing the two mixtures they came from. The problem that arises is that the constant term ($q_0$ in Equation~\ref{eqn:equivalence}) is exactly the probability of no coalescence, which is fixed by the subpopulation sizes, and therefore fixed by the desired mixture. Therefore, if we are not careful, the probability of no coalescence will be significantly different between our two population histories, making them easily distinguishable. 

\begin{proof}[Proof of Theorem~\ref{thm:pop-reduction}]
Consider the following population history for $R$:
\begin{itemize}
\item The first $k$ populations of $R$ are constructed according to the proof of Theorem~\ref{thm:equivalence} from the mixture $P$. Let $N_i$ denote the size if the $i^{th}$ population in the more recent period $[0, t_0)$.
\item The $(k+1)^{th}$ and $(k+2)^{th}$ populations of $R$ will have sizes $N_{k+1}, N_{k+2}$ in the period $[0, t_0)$, such that $N_{k+1}^2 + N_{k+2}^2 = 64 \max(k^2, \ell^2) \sum_{i=1}^{k} N_i^2$. The exact sizes of $N_{k+1}$ and $N_{k+2}$ will be set later.
\item Every pair of samples in populations $k+1$ and $k+2$ coalesce in the closed interval $[0, t_0]$. This can be achieved by having populations $k+1$ and $k+2$ shrink to size ``zero'' at time $t_0$, so any pair of individuals in these populations who have not coalesced in the interval $[0, t_0)$ will coalesce at time $t_0$.
\end{itemize}
Similarly, let the population history $S$ be constructed from the mixture $Q$ in the same way (substituting $Q$ for $P$ and $\ell$ for $k$ in what appears above), and let $M_j$ denote the size of $j^{th}$ population of the population history $S$ in the (most recent) period $[0, t_0)$ (where $j$ ranges from 1 to $\ell + 2$). Recall that we want to establish the following four properties:
\begin{enumerate}
\item $R$ has $k+2$ subpopulations and $S$ has $\ell + 2$ subpopulations.
\item $ R[T > t + t_0 | T\neq \infty, T > t_0] = P(T > t)$ and $ S[T > t + t_0 | T\neq \infty, T > t_0] = Q(T > t)$
\item $R[T = t_0] = S[T = t_0]$. 
\item $R[T = \infty] = S[T = \infty]$.
\end{enumerate}

Note that the coalescence probability for the population history $R$ at exactly $t_0$ is 
\[R[T = t_0] = \frac{1}{(\sum_{i = 1}^{k + 2} N_i)^2} \left(N_{k+1}^2 e^{-t_0 / N_{k+1}} + N_{k+2}^2 e^{-t_0 / N_{k+2}}\right). \]
First, we analyze what happens in the limit $t_0 \to 0$, disregarding the term that comes from the probability of coalescing before time $t_0$, so the equations simplify to 
\begin{equation}
\label{eq:t0-irrelevant-1}
R[T = t_0] = \frac{N_{k+1}^2 + N_{k+2}^2}{(\sum_{i = 1}^{k + 2} N_i)^2}
\end{equation}
and similarly
\begin{equation}
\label{eq:t0-irrelevant-2}
S[T = t_0] = \frac{M_{k+1}^2 + M_{k+2}^2}{(\sum_{i = 1}^{\ell + 2} M_i)^2}
\end{equation}
We do this mainly for simplicity, and 
at the end show how to handle the general case, supposing $t_0$ is sufficiently small. 

It is clear that $R$ and $S$ satisfy Property (1) of Theorem~\ref{thm:pop-reduction} by construction, as they have the correct numbers of subpopulations. Property (2) is similarly satisfied by construction and by an application to Theorem~\ref{thm:equivalence}. Indeed, when we sample from the history $R$ conditioned on coalescence before time $\infty$ and no coalescence at or before time $t_0$, we can see that the two individuals sampled must be from one of subpopulations $1$ to $k$ (since any pair of individuals sampled from subpopulations $k+1$ or $k+2$ coalesce in the interval $[0,t_0]$). Repeating this argument for $S$ establishes property (2). 

It remains to show that Properties (3) and (4) hold, and for this we will need to set $N_{k+1}, N_{k+2}, M_{\ell+1}$, and $M_{\ell + 2}$ appropriately. In order to do this, we first rewrite Properties (3) and (4) using Equations~\ref{eq:t0-irrelevant-1} and~\ref{eq:t0-irrelevant-2}, to see what conditions we need to satisfy with $N_{k+1}, N_{k+2}, M_{\ell+1}$, and $M_{\ell + 2}$. 

Noting that the probability of coalescence at \emph{exactly} time $t_0$ from population history $R$ (resp. $S$) is exactly the probability of sampling a pair of individuals from $N_{k+1}$ or $N_{k+2}$ (resp. $M_{\ell+1}$ or $M_{\ell+2}$), we get that $R$ and $S$ will satisfy Property (3) if and only if the population sizes satisfy the following (desired) equation:
\begin{equation}
\label{eqn:pop-reduction-1}
\frac{N_{k+1}^2 + N_{k+2}^2}{\left(\sum\limits_{i=1}^{k+2}N_i\right)^2} = 
\frac{M_{\ell+1}^2 + M_{\ell+2}^2}{\left(\sum\limits_{j=1}^{\ell+2}M_j\right)^2}
\end{equation}
Similarly, $R$ and $S$ will satisfy Property (4) if and only if the population sizes satisfy the following (desired) equation:
\begin{equation}
\frac{\sum\limits_{j=1}^{\ell+2} M_j^2}{\left(\sum\limits_{j=1}^{\ell+2} M_j\right)^2} = \frac{\sum\limits_{i=1}^{k+2} N_i^2}{\left(\sum\limits_{i=1}^{k+2} N_i\right)^2} \label{eqn:pop-reduction-2}
\end{equation}

So it suffices to describe how to set $N_{k+1}, N_{k+2}, M_{\ell+1}$, and $M_{\ell+2}$ such that Equations~\ref{eqn:pop-reduction-1} and~\ref{eqn:pop-reduction-2} are satisfied. First, we rearrange Equations~\ref{eqn:pop-reduction-1} and~\ref{eqn:pop-reduction-2} to get the equivalent (desired) set of equalities
\begin{equation}
\label{eqn:pop-reduction-3}
\frac{\left(\sum\limits_{i=1}^{k+2}N_i\right)^2}{\left(\sum\limits_{j=1}^{\ell+2}M_j\right)^2} = \frac{N_{k+1}^2 + N_{k+2}^2}{M_{\ell+1}^2 + M_{\ell+2}^2} = \frac{\sum\limits_{i=1}^{k+2} N_i^2}{\sum\limits_{j=1}^{\ell+2} M_j^2}
\end{equation}
Note that the second equality is trivially satisfied by our stipulations that $N_{k+1}^2 + N_{k+2}^2 = 64 \max(k^2, \ell^2) \sum_{i=1}^{k} N_i^2$ and $M_{\ell+1}^2 + M_{\ell+2}^2 = 64 \max(k^2, \ell^2) \sum_{i=1}^{\ell} M_j^2$. 

Now we move on to describing how to set population sizes to satisfy the first equality. Note that we can write 

\begin{equation} 
\label{eqn:pop-reduction-4}
\left(\sum_{j=1}^{\ell+2}M_j\right)^2 = ((M_{\ell+1} + M_{\ell+2}) + \sum_{j=1}^{\ell}M_j)^2
\end{equation}
Note that the RHS of Equation~\ref{eqn:pop-reduction-4} is a continuous and strictly increasing function of $M_{\ell+1} + M_{\ell+2}$, and furthermore that we can set $M_{\ell+1} + M_{\ell+2}$ to be anywhere in the range \\$\left[\sqrt{M_{\ell+1}^2 + M_{\ell+2}^2}, \sqrt{2} \cdot \sqrt{M_{\ell+1}^2 + M_{\ell+2}^2}\right]$ while keeping $M_{\ell+1}^2 + M_{\ell+2}^2$ constant. Furthermore, note that 

\begin{align}
\left(\sum_{j=1}^{\ell+2}M_j\right)^2 &= ((M_{\ell+1} + N_{\ell+2}) + \sum_{j=1}^{\ell}M_j)^2 \nonumber \\
&\leq ((9/8)(M_{\ell+1} + M_{\ell+2}))^2 \nonumber \\
&= \frac{81}{64}(M_{\ell+1} + M_{\ell+2})^2 \label{eqn:pop-lb-M}
\end{align}
where the second line uses the fact that we set $M_{\ell+1}^2 + M_{\ell+2}^2$ sufficiently large such that
\[ M_{\ell+1} + M_{\ell+2} \geq \sqrt{64\ell^2 \sum_{j=1}^{\ell}M_j^2} \geq \sqrt{64 \left(\sum_{j=1}^{\ell} M_j\right)^2} \geq 8\sum_{j=1}^{\ell} M_j\]

We can do the same bound for $\left(\sum_{i=1}^{k+2}N_i\right)^2$, and hence we get that the LHS of Equation~\ref{eqn:pop-reduction-3} can be bounded on both sides as follows
\begin{equation} 
\label{eq:pop-reduction-5}
\frac{64}{81}\frac{(N_{k+1} + N_{k+2})^2}{(M_{\ell+1} + M_{\ell+2})^2} \leq \frac{\left(\sum\limits_{i=1}^{k+2}N_i\right)^2}{\left(\sum\limits_{j=1}^{\ell+2}M_j\right)^2} \leq \frac{81}{64}\frac{(N_{k+1} + N_{k+2})^2}{(M_{\ell+1} + M_{\ell+2})^2}
\end{equation}

Now suppose that we initially set $N_{k+2} = M_{\ell+2} = 0$, and suppose that this gives us that
\begin{equation} 
\label{eq:pop-reduction-6}
\frac{N_{k+1}^2 + N_{k+2}^2}{M_{\ell+1}^2 + M_{\ell+2}^2} = \alpha \cdot \frac{\left(\sum\limits_{i=1}^{k+2}N_i\right)^2}{\left(\sum\limits_{j=1}^{\ell+2}M_j\right)^2} 
\end{equation}
for some $\alpha < 1$. We know by $\alpha \geq \frac{64}{81}$ by noting that $N_{k+2} = M_{\ell+2} = 0$ and applying the upper bound of Equation~\ref{eq:pop-reduction-5}. 

We can continuously increase $M_{\ell+1} + M_{\ell+2}$ while keeping $M_{\ell+1}^2 + M_{\ell+2}^2$ constant by moving $M_{\ell+1}$ and $M_{\ell+2}$ relatively closer together, until we satisfy Equation~\ref{eq:pop-reduction-6} with $\alpha=1$, satisfying the first equality in Equation~\ref{eqn:pop-reduction-3}. This is because we can increase $(M_{\ell+1} + M_{\ell+2})^2$ by a factor of up to 2 overall in this manner, and using the fact that 
\[ (M_{\ell+1} + M_{\ell+2})^2 \leq \left(\sum_{j=1}^{\ell+2}M_j\right)^2 \leq \frac{81}{64}(M_{\ell+1} + M_{\ell+2})^2\]
we conclude that we can increase $\left(\sum_{j=1}^{\ell+2}M_j\right)^2$ by at least a multiplicative factor of $2 \cdot (64/81) \geq (81/64)$ in this fashion. Hence, there exists a setting of $M_{\ell+1}$ and $M_{\ell+2}$ that satisfies both equalities in Equation~\ref{eqn:pop-reduction-3}.

To handle the case where $1 < \alpha \leq 81/64$, we instead continuously increase $N_{k+1} + N_{k+2}$ while keeping $N_{k+1}^2 + N_{k+2}^2$ fixed, and the argument goes through mutatis mutandis.

Finally, we describe how to handle the general case where $t_0 > 0$ and $t_0$ is sufficiently small. We can ensure that property (2) by using the same rescaling method as in the proof of Theorem~\ref{thm:equivalence}. It remains to show that properties (3) and (4) hold by modifying $N_{k + 1}, N_{k + 2}, M_{\ell + 1}, M_{\ell + 2}$ appropriately. Recall that in this case, $R[T = t_0] = \frac{1}{(\sum_{i = 1}^{k + 2} N_{i})^2}\left( N_{k+1}^2 e^{-t_0 / N_{k+1}} + N_{k+2}^2 e^{-t_0 / N_{k+2}}\right)$ and similarly for $S$.
As before, property (3) holds by fixing the ratio of $N_{k + 1}^2 e^{-t_0 / N_{k + 1}} + N_{k + 2}^2 e^{-t_0/N_{k + 2}}$ and $M_{k + 1}^2 e^{-t_0/M_{k + 1}} + M_{k + 2}^2 e^{-t_0/M_{k + 2}}$. As before, this leaves a degree of freedom in the value of $M_{k + 1} + M_{k + 2}$ which we use to guarantee property (4) holds assuming $t_0$ is sufficiently small.
\end{proof}

%% file: info-lower-bound-revised.tex
\section{Information-Theoretic Lower Bounds} \label{sec:info-lower-bound}

We will show that an exponential dependence on the number of 
components is not just a limitation of the matrix pencil method,
but is in fact information-theoretically necessary. Furthermore
we will show that our dependence on the gap parameter $\Delta$ is also
information-theoretically necessary, proving that the Matrix pencil method
has essentially the correct dependence on all parameters.

Our constructions rely on interpreting polynomials as the difference between two mixtures of exponentials, under the re-parametrization $x \mapsto e^{-x}$.
From this point of view, and in light of our upper bound (Theorem~\ref{thm:matrix-pencil-sample-complexity}), the hardest examples should come from families of polynomials whose supremum on $[0,1]$ are exponentially small in the number of terms.

Of particular interest are two families of polynomials -- one attributed to Chebyshev and the other to Tura\'n \cite{nazarov1993local, turan1984new}.
These polynomials are tightest examples of extremal families, in the sense that they have the largest growth rate outside of the unit interval.
In particular, the Chebyshev family are tight examples of the Remez Inequality \cite{remez1936propriete} which is a bound on the supremum norm of a polynomial, while Tura\'n's polynomial family (essentially) serve as tight examples for the closely related Nazarov-Tura\'n lemma (\cref{thm:nazarov-turan}).
Notably, the latter result is specialized to linear combinations of exponential functions -- precisely the setting that we are working in, which make Tura\'n's polynomials fitting candidates for constructing matching lower bounds.


\subsection{A lower bound construction}

\subsubsection{Tur\'an's polynomials and their properties}
\label{sec:turan}
First, we recall the construction by Tur\'an's family of polynomials.
As mentioned earlier in this section, a slight multivariate generalization of this construction is central to the proof of Tur\'an's First Main Theorem \cite{turan1984new} and to ``Tur\'an's proof'' of the Nazarov-Tur\'an Lemma \cite{nazarov1993local}. 

\begin{defn}[Tur\'an Polynomials]
Fix positive integers $m$ and $n$ such that $m > n$. 
The $(m,n)$-th Tura\'n polynomial $Q_{m,n}(z)$ is defined to be $(1-z)^{n} \sigma_{m}(z)$, where $\sigma_{m}(z) = \sum_{i=0}^{m} a_i z^i$ is the degree-$m$ truncation of the power series expansion of $(1-z)^{-n}$.
\end{defn}
These polynomials have several very useful properties; we recall these facts
below, along with proofs for the reader's convenience.
First we recall the following basic fact:
\begin{proposition}
\label{prop:partial-sum}
For any $k$ and $n$,
\[ [z^k] \frac{1}{(1 - z)^n} = \binom{n+k-1}{n-1} \]
where the left hand side denotes the coefficient of $z^k$ in the power series
expansion about $0$.
\end{proposition}
\begin{proof}
Consider the power series expansion of $1 / (1-z)$ as $1 + z + z^2 + \cdots$. It follows that the $k$th coefficient of the expansion of $\left(\frac{1}{1-z}\right)^n$ is exactly the number of ways in which we can select $n$ nonnegative numbers to sum up to $k$, which is $\binom{n+k-1}{n-1}$.
\end{proof}
From this one can verify that the Turan polynomials has several
very interesting and useful properties.
\begin{lemma}\label{lem:turan-polynomials}
Let $n,m$ be positive integers such that $m > n$, and let $Q = Q_{m,n}$.
Then
\begin{enumerate}
\item $Q(1) = 0$ so the sum of the coefficients equals $0$.
\item The only nonzero monomials in $Q(z)$ are the constant term $z^0$ and $z^{m + 1}, \ldots, z^{m + n}$.
\item The leading coefficient has norm at least $\left(\frac{m + n - 1}{n - 1}\right)^{n - 1}$.
\item $Q(z) \in [0,1]$ for $z \in [0,1]$.
\item The signs of the coefficients of $Q(z)$ are alternating. 
\end{enumerate}
\end{lemma}
\begin{proof}	
The first point is immediate.
For Property (2), observe that degree of $Q$ is at most $m+n$ and furthermore that, for $0 < k \le m$,
\[ [z^k] Q(z) = [z^k] \sigma_m(z) (1 - z)^n = [z^k] \frac{(1 - z)^n}{(1 - z)^n} = 0. \]
For Property (3), we use Proposition~\ref{prop:partial-sum} to observe that
\begin{align*}
	|[z^{m + n}] Q(z)| 
	&= [z^m] \frac{1}{(1 - z)^n} = \binom{m + n - 1}{n - 1} \\
	&= \left(1 + \frac{m}{n - 1}\right)\left(1 + \frac{m}{n - 2}\right) \cdots 
	\left(1 + \frac{m}{1}\right) \ge \left(1 + \frac{m}{n - 1}\right)^{n - 1}.
\end{align*}
For Property (4), observe that the coefficients of the power series of $\left(\frac{1}{1 - z}\right)^n$
are all nonnegative, so for $z \in [0,1)$
\[ 0 \le (1 - z)^n \sigma_m(z) < (1 - z)^n \frac{1}{(1 - z)^n} \le 1. \]
and furthermore that $Q(1) = 0$.
Finally, to infer Property (5), observe that $Q$ has exactly $n+1$ monomials and a positive root of multiplicity $n$ at $z=1$. Apply Descartes' Rule of Signs, which tells us that the number of sign changes is equal to the number of positive roots with multiplicity.
\end{proof}

\subsection{Hard exponential mixtures from Tur\'an polynomials}

We remind ourselves of some tools from information theory.
For a pair of probability measures $P$ and $Q$ corresponding to densities $p$ and $q$, their \emph{$\chi^2$-divergence} is defined by
\[ \chi^2(P,Q) := \int \frac{(p(x) - q(x))^2}{q(x)} dx. \]
We recall the following well-known fact relating Hellinger and $\chi^2$-divergence; we include the proof for the reader's convenience.
\begin{proposition}
  \[ H^2(P,Q) \le \frac{1}{2} \chi^2(P,Q) \]
\end{proposition}
\begin{proof}
  Directly, we have
  \[ H^2(P,Q) = \frac{1}{2} \int (\sqrt{q(x)} - \sqrt{p(x)})^2 dx = \frac{1}{2} \int \frac{(q(x) - p(x))^2}{(\sqrt{q(x)} + \sqrt{p(x)})^2} \le \frac{1}{2} \int \frac{(q(x) - p(x))^2}{q(x)} dx.  \]
\end{proof}
This fact will be useful in obtaining the following estimate, which is tight in terms of the exponent of the gap parameter $\Delta$.
\begin{theorem}\label{thm:turan-lower-bound-h2}
For any positive integers $m,n$ such that $m > 2n$, let $\Delta = 1/(m + 2n + 1)$.
Fix $\alpha \in (0,1/2)$.
Then there exists two mixtures of exponentials $P_1$, $P_2$ with CDFs defined on $[0,\infty)$ by
\[ 1 - F_1(t) = (1 - \alpha + a_0) e^{-\nu t} + \sum_{j=1}^{n} a_j e^{-\lambda_j t} \]
and
\[ 1 - F_2(t) = (1 - \alpha) e^{-\nu t} + \sum_{j=1}^{n} b_j e^{-\mu_j t} \]
such that
\begin{enumerate}
\item (Normalization) $\nu,\lambda_j, \mu_j \in (0,1]$.
\item (Exponents are well-separated) All of the elements of the set $\nu \cup \{\lambda_j\} \cup \{\mu_j\}$ are separated by at least $\Delta$.
Furthermore, the sets $\{ \nu\}, \{\lambda_j\},$ and $\{ \mu_j\}$ are disjoint and interlaced, e.g. $\nu < \lambda_1 < \mu_1 < \lambda_2 < \mu_2 < \cdots < \lambda_n < \mu_n.$
\item (Coefficients are bounded) $\sum_j a_j = \alpha$ and $\sum_j b_j = \alpha$. 
\item (Indistinguishability) $H^2(P_1, P_2) \le \frac{\alpha^2}{(2n - 1)^2}\left[\Delta (2n - 1)\right]^{4n - 4}$.
\end{enumerate}
\end{theorem}
\begin{proof}
Let $Q_{m,2n}$ be the $(m,2n)$-th Tura\'n polynomial.
We start by re-centering it so that its average value over the unit interval is zero, by writing
\[ R_{m,2n}(x) := Q_{m,2n}(x) - \int_0^1 Q_{m,2n}(y) dy. \]
Let $C_{m,2n}$ be the sum of the positive entries of $Q_{m,2n}$.
By Property 3 of Lemma~\ref{lem:turan-polynomials}, we have $C_{m,2n} \ge [(m + 2n - 1)/(2n - 1)]^{2n - 1}$. Integrate $R_{m,2n}$ to give the polynomial
\[ S_{m,2n}(x) := \int_0^x R_{m,2n}(y) dy \]
and observe that $S_{m,2n}(0) = 0$ and $S_{m,2n}(1) = \int_0^1 Q_{m,2n}(y) dy - \int_0^1 Q_{m,2n}(y) dy = 0$. 
Define $C'_{m,2n}$ to be the sum of the positive entries of $S_{m,2n}(x)$.
Observe, by the power rule for integrating monomials, that
\[ C'_{m,2n} \ge C_{m,2n}/(m + 2n + 1) \ge \frac{ [(m + 2n - 1)/(2n - 1)]^{2n - 2} }{ 4n-2 }. \]
Define $f(x) = \alpha S_{m,2n}(x)/C'_{m,2n}$, so
\[ f(x) = f_1 x + \sum_{k = m + 2}^{m + 2n + 1} f_k x^k \]
and the sum of all of the positive (or all of the negative) coefficients of $f$ is
$\alpha$. 
The coefficients here will be exactly the $a_j$'s and $b_j$'s, which verifies property (3).

We now split the polynomial into a positive part and a negative part.
Define $f_+(x) = (1 - \alpha + f_1) x + \sum_{k = m + 2; f_k > 0}^{m + 2n + 1} f_k x^k$ and 
$f_{-}(x) = (1 - \alpha) x + \sum_{k = m + 2; f_k < 0}^{m + 2n + 1} (-f_k) x^k$ so that
\[ f(x) = f_+(x) - f_-(x) \]
and define CDFs of mixtures of exponentials $F_1,F_2$ by
\[ F_1(t) := 1 - f_+(e^{-t}), \qquad F_2(t) := 1 - f_-(e^{-t}). \]
Observe by the chain rule that
\[ p_1(t) := \frac{d}{dt} F_1(t) = f'_+(e^{-t}) e^{-t}, \qquad p_2(t) := \frac{d}{dt} F_2(t) = f'_-(e^{-t}) e^{-t} \]
so by linearity of the derivative,
\[ p_1(t) - p_2(t) = e^{-t}(f'_+(e^{-t}) - f'_-(e^{-t})) = e^{-t} f'(e^{-t}). \]
Observe that for $x \in [0,1]$
\[ |f'(x)| = \frac{\alpha}{C'_{m,2n}} |S'_{m,2n}(x)| = \frac{\alpha}{C'_{m,2n}} |R_{m,2n}(x)| = \frac{\alpha}{C'_{m,2n}} \left|Q_{m,2n}(x) - \int_0^1 Q_{m,2n}(y) dy\right| \le \frac{\alpha}{C'_{m,2n}}. \]
In the last step, we used Property 4 of Lemma~\ref{lem:turan-polynomials}; the quantity inside the absolute value is the difference between $Q_{m,2n}$ and its average over $[0,1]$, which implies that the difference is bounded in absolute value by $1$.

Finally, notice that $f'_-(x) \ge ( 1- \alpha) \ge 1/2$ for all $|x| \ge 0$.
If $P_1$ and $P_2$ are measures having densities $p_1$ and $p_2$, then
\begin{align*}
\chi^2(P_1,P_2)
  = \int_t \frac{(p_1(t) - p_2(t))^2}{p_2(t)} dt
  &= \int_t e^{-t} \frac{(f'_+(e^{-t}) - f'_-(e^{-t}))^2}{f'_-(e^{-t})} dt \\
  &\le \int_t e^{-t} \frac{(\alpha/C'_{m,2n})^{2}}{(1/2)} dt \\
  &\le 2\alpha^2 (4n-2)^2 \left[\frac{2n - 1}{m + 2n - 1}\right]^{4n - 4}.
\end{align*}
Therefore,
\[ 
	H^2(P_1,P_2) 
	\le \frac{1}{2} \chi^2(P_1,P_2) 
	\le \alpha^2 (4n-2)^2 \left[\frac{2n - 1}{m + 2n - 1}\right]^{4n - 4} 
	\le \alpha^2 (4n-2)^2 \left[\Delta (2n - 1)\right]^{4n - 4}.
\]
This computation verifies condition (4).
Finally, we re-scale $t$ by making the transformation $t \mapsto \frac{t}{m + 2n + 1}$.
By Properties 2 and 5 of Lemma~\ref{lem:turan-polynomials}, these mixtures satisfy (1) and (2).
\end{proof}

\begin{corollary} \label{cor:turan-tv}
  Let $P_1$ and $P_2$ be as in \cref{thm:turan-lower-bound-h2}, $\alpha \in (0,1/2)$ be arbitrary, and let $k = n + 1$ be the number of components in each of the two mixtures.
  Let $P^{\otimes N}$ denote the product measure corresponding to taking $N$ iid samples from probability measure $P$.
  Then
	\[ 
		\TV(P_1^{\otimes N}, P_2^{\otimes N}) \le \frac{\alpha \sqrt{2N} }{2k - 3} \left[\Delta(2k - 3)\right]^{2k - 4}.
	\]
\end{corollary}
\begin{proof}
  We use the comparison inequality between TV and $H^2$, together with the tensorization inequality for $H^2$ to conclude
  \[ \TV(P_1^{\otimes N}, P_2^{\otimes N}) \le \sqrt{2 H^2(P_1^{\otimes N}, P_2^{\otimes N})} \le \sqrt{2 N H^2(P_1, P_2)} \]
  and then apply \cref{thm:turan-lower-bound-h2}.
\end{proof}

Note that if we take $m$ to be much larger than $k$, $\Delta$ is much smaller than $2k$.
Informally speaking, Corollary~\ref{cor:turan-tv} translates to a sample complexity lower bound of 
\[
	\Omega \left( \frac{k^2}{\alpha^2} \left( \frac{1}{\Delta} \right)^{4k-4} \right)
\]
in the hypothesis testing problem of distinguishing between $P_1$ and $P_2$ for any choice of $\alpha$.
Compare this to Equation~\ref{eq:nt-error} of 
\cref{thm:nt-hypothesis-testing}, which gives an upper bound of 

\[
O \left(\frac{1}{p_{leading}^2}(c_{\Delta}/\Delta)^{4k - 2} \log(2/\alpha)\right)
\]
where $c_{\Delta} =  8e^2 / \left( \min(1/\Delta_{leading},(2k-1)) \right)$
for simple-versus-composite hypothesis testing at a constant significance level.
This shows the matching the quadratic dependence on the coefficients as well as the exponent $4k$ in the gap parameter.

For comparison, we also juxtapose the result from \cref{thm:matrix-pencil-sample-complexity}, which gives an upper bound of
\[
	O \left( 
    \frac{k^{7}}{p_{\min}^4 } \left( \frac{2e}{\Delta} \right)^{4k} 
    \right)
\]
for the learning problem, which gives a quartic dependence on the coefficients instead. We suspect this difference is not simply an artifact of the analysis; in practice, the convex programming approach seems to succeed with fewer samples than the MPM.

Finally, we can also give a lower bound on the minimax rate of learning the coefficients of a mixture of exponentials:

\begin{theorem}\label{thm:minimax-lower-bound}
Let $k > 3$, $m > k$ and $\Delta = 1/(m + k - 1)$.
There exists $\lambda_1,\ldots,\lambda_{k}$ which are $\Delta$-separated such that for any estimator $\hat{p}$ of the coefficients from $N$ samples of the mixture of exponentials with CDF $F(t) = 1 - \sum_j p_j e^{-\lambda_j t}$:
\[ 
	\inf_{\hat{p}} \max_{p} \E_{p} \|p - \hat{p}\|_1 
    \ge 
    \frac{1}{4} \min\left( 
    	1,
        \frac{k-3}{\sqrt{2N}} \left( \frac{1}{\Delta(k-3)} \right)^{k-4}
    \right) 
\]
\end{theorem}
\begin{proof}
This is derived by the usual reduction to hypothesis testing. Let $\alpha \in (0,1/2)$ to be fixed later, and let $P_1$ and $P_2$ be the mixtures of exponentials from \cref{cor:turan-tv} with coefficients $\{a_i\}$ and $\{b_j\}$ (supported at disjoint indices) with parameter $\alpha$.
As a reminder, \cref{cor:turan-tv} gives (where we write $k'$ as a placeholder for $k$)
\[
	\TV(P_1^{\otimes N}, P_2^{\otimes N})
    \leq
    \frac{\alpha\sqrt{2N}}{2k'-3} \left( \Delta(2k'-3) \right)^{2k'-4}.
\]
Now let
\[
	\alpha := \min\left(
    	\frac{1}{2},
        \frac{2k-3}{2\sqrt{2N}} \left( \frac{1}{\Delta(2k'-3)} \right)^{2k'-4}
    \right)
\]
so $\TV(P_1, P_2) \le 1/2$.
This means we can couple the distributions so that they draw identical outputs with probability at least $1/2$.
Let $k = 2k'$.
Under such a coupling, for at least one of $P_1$ or $P_2$, its coefficients $p$ are far from $\hat{p}$:
\begin{align*}
	\| \hat{p} - p \|_1
    \geq
    \frac{1}{2} (a_0 + \sum_{j=1}^{k} |a_j - b_j|)
    =
    \alpha
\end{align*}
with probability at least $1/4$, so the expected error is at least $\alpha/4$. 
For the last equality in the above expression, we again keep in mind that the coefficients $a_j$ and $b_j$ are supported on disjoint sets of indices.
\end{proof}

\subsection{Equally-spaced exponents via Chebyshev Polynomials}

In this section, we briefly discuss a slightly different family of polynomials, which are classically attributed to Chebyshev.
It will turn out that they give a similar exponential-type bound for learning the $\lambda_j$'s when the gap $\Delta$ is as large as possible.
We remind ourselves of their definition now:
\begin{defn}
The $k$th Chebyshev polynomial (of the first kind) $T_k$ is given by the recursive relation
\[
	T_{n+1}(x) = 2xT_n(x) - T_{n-1}(x)
\]
starting with
$T_0(x) = 1$ and
$T_1(x) = x$.
\end{defn}

\noindent
These polynomials have the property that for any $k$,
\begin{enumerate}
	\item $|T_k(x)| \leq 1$ for all $x \in [-1,1]$.
    \item If $k$ is even, the monomials that appear in $T_k$ are all even powers of $x$ up to $k$. If $k$ is odd, then the monomials are all odd powers of $x$ up to $k$.
    \item The coefficient of $x^k$ is $2^{k-1}$ for $k \geq 1$.
\end{enumerate}
Note that Property 3 of $T_k$ is comparable to Property 3 of $Q_{m,n}$ in Lemma~\ref{lem:turan-polynomials}.
Of particular interest here is Property 2 of $T_k$, which states that the exponents that appear are equally spaced.
Indeed, the same technique of normalizing then re-weighting the coefficients by an arbitrary $\alpha \in (0,1/2)$ used in the proof of Theorem~\ref{thm:turan-lower-bound-h2} yields the following lower bound:
\begin{theorem}
\label{thm:chebyshev-1}
Fix any positive integer $k$, and fix $\alpha \in (0,1/2)$.
There exist two mixtures of exponentials $P_1$, $P_2$ with CDFs defined on $[0,\infty]$ by
\[ 
	1 - F_1(t) 
    = 
    (1-\alpha + a_0)e^{-\nu t} + \sum_{j=1}^{k} a_j e^{-\lambda_j t} 
\]
and
\[
	1 - F_2(t) 
    = 
    (1 - \alpha) e^{-\nu t} + \sum_{j=1}^{k} b_j e^{-\mu_j t}
\]
such that
\begin{enumerate}
\item (Normalization) $\nu, \lambda_j, \mu_j \in (0,1]$.
\item (Well-separated exponents) All of the elements of the set $\nu \cup \{\lambda_j\} \cup \{\mu_j\}$ are separated by at least $1/(2k)$.
\item (Bounded coefficients) $\sum_j \alpha_j = \alpha$ and $\sum_j b_j = \alpha$.
\item (Indistinguishability) $H^2(P_1,P_2) \le \frac{\alpha^2 (2k+1)^2}{2^{4k-2}}$.
\end{enumerate}
\end{theorem}

%% file: population-algorithm.tex

\section{A reconstruction algorithm for population history} \label{sec:population-history}

In this section, we describe an algorithm that takes as input $L$ independent and identically distributed 2-sample coalescence times $c_1, \ldots c_L$ and outputs a population shape.
In particular, we assume that each sample is an accurate measurement of coalescence times, with the interpretation that they are i.i.d. from the distribution of $T$ (\cref{sec:coalescent}).

Our strategy is reminiscent of the algorithm provided by \cite{kim2015can} for the $1$-component case.
We will be reconstructing the history iteratively, starting from the most recent event and going as far backwards in time as possible.
The hope is that if we have a guarantee that we have accurate constructions of the subpopulation structures in the intervals $I_1, \ldots I_j$, then we should be able to provide a good guess of what happened at the boundaries between $I_{j-1}$ and $I_j$ at time $t_{j-1}$.

\subsection{Preliminary ideas} \label{sec:reconstruction-lems}
\subsubsection{Inferring events -- split, merge or size change?}
Assume, for the sake of argument, that we are given the exact population shape $\vN$ for the interval $I_{j-1}$.
Since our model only allows for up to one merge or one split at any given time, the only possibilities are $D_{j} \in \{D_{j-1} - 1, D_{j-1}, D_{j-1} + 1\}$.
Therefore, the maximal number of components in the exponential mixture in $I_j$ including the constant term is $D_{j-1}+2$.

In the noiseless scenario, we can run the matrix pencil method with component size $k = D_{j-1}+2$ on estimates $v_0, \ldots, v_{2D_{j-1}+3}$ describing CDF values on $I_j$.
In this case, we get some collection of $k$ eigenvalues $\alpha_0 > \cdots > \alpha_{D_{j-1}+1}$.
If in reality $D_{j} < D_{j-1}+1$, then $0$ must be an eigenvalue.
Therefore, we can discard $\alpha_i$ that are equal to zero to get the true collection of components.

In the noisy case, by \cref{thm:matrix-pencil-sample-complexity}, given a sufficiently large collection of $L$ iid samples that lie in $I_j$ and to its right, we can discard $\alpha_i$ that are close to zero: those that are smaller than some error threshold\footnote{Note that we provide the exact constants in the full proof of \cref{thm:matrix-pencil-sample-complexity} found in \cref{sec:matrix-pencil-analysis}.}, $\eta_1 = \Theta\left( \left(\frac{2e}{\Delta} \right)^{2k} \sqrt{\frac{1}{L}} \right)$.
Just as in \cref{thm:matrix-pencil-sample-complexity}, $\Delta$ is a lower bound on the gap between any two distinct exponents in the model. 
For the reconstruction algorithm, it is a free parameter that determines the number of samples required.
Conversely, it is the threshold for learning given the number of samples available.

Once this is done, all that's left is to semantically link the subpopulations in $I_{j-1}$ to those of $I_{j}$.
While it is tempting to simply index subpopulations in decreasing order of $\alpha$'s and say that corresponding subpopulations are matching, the labeling becomes inconsistent, for instance, under consideration of large size changes or under merges/splits.
We address this issue in the upcoming section, \cref{sec:reconstruction-recursion}.

\subsubsection{The recursion}\label{sec:reconstruction-recursion}
We briefly recall the model derived in \cref{sec:model-derivation}. 
If the population history $\vN$ is constant in the interval $I = [a,b]$, then $T$ satisfies
\[
	\Pr(T > a+t \;|\; T > a) = p_0 + \sum_{\ell=1}^{D} p_{\ell} e^{-\lambda_\ell t},
\]
where for each $\ell \geq 1$, $p_\ell := \Pr(\cE_{\ell \ell} \;|\; T > a)$.

Since our algorithm works iteratively, it would be helpful to relate the exponential mixture parameters of the interval $I_j$ to those of $I_{j-1}$, assuming $\vN$ is constant in $I_j = [b,c]$ and in $I_{j-1} = [a,b]$, for $a < b < c$.
\emph{The overall goal here is to describe a ``matching'' scheme for subpopulations of $I_{j-1}$ to those of $I_{j}$.}
For the sake of brevity, we introduce some notation.
\begin{itemize}
\item Let $n = D_j$ and $m = D_{j-1}$.
\item For $i = 1$ to $D_j$, let $\cE_{i}^{j}$ denote the event that both lineages trace back to subpopulation indexed $i$ in the interval $I_j$.
\item Let $\{ (p_i, \alpha_i) \}_{i=1}^{n}$ denote the coefficients and exponentials ($\alpha_i = e^{-\lambda_i}$) for the interval $I_j$.
\item Let $\{ (q_i, \beta_i) \}_{i=1}^{m}$ denote the analogous parameters for $I_{j-1}$.
\end{itemize}

We only need to consider subpopulations $i$ in $I_j$ which is linked to subpopulation $i'$ in $I_{j-1}$, in the absence of splits or merges specifically involving $i$ or $i'$.
Indeed, since we only allow one split or merge at a time, either of these events can be inferred from the fact that all but at most two subpopulations from $I_j$ can be ``matched'' with subpopulations in $I_{j-1}$ if a split/merge occurred.

Observe that $\cE_{i}^j = \cE_{i'}^{j-1}$, and that $\{T > b\} \subset \{T > a\}$ as events.
We derive the following:
\begin{align*}
	p_{i} 
    &= \Pr(\cE_{i}^j \;|\; T > b) \\
    &= \Pr( \cE_{i'}^{j-1} \land \{T > a\} \;|\; T > b ) \\
    &= \Pr(T > b \;|\; \cE_{i'}^{j-1} \land \{T > a\} ) \cdot
    \frac{ \Pr(\cE_{i'}^{j-1} \land \{T > a\} )}{\Pr(T > b)} \\
    &= \Pr(T > b \;|\; \cE_{i'}^{j-1} \land \{T > a\} ) \cdot \Pr(\cE_{i'}^{j-1} \;|\; T > a) \cdot \frac{\Pr(T > a)}{\Pr(T > b)} \\
    &= \left( \beta_{i'}^{b-a} \right) \cdot \left( q_{i'} \right) \cdot \frac{1}{\Pr(T > b \;|\; T > a)}
\end{align*}
In particular, if we can accurately estimate $\hat{p}, \hat{\alpha}, \hat{q}, \hat{\beta}$ from coalescent samples $t_1,\ldots,t_L$, we can infer that population $i$ links to population $i'$ if
\begin{equation}\label{eq:coeff-approx}
	\hat{p}_i \cdot \frac{\#\{ k : t_k > b\}}{\#\{ k : t_k > a\}}
    \approx
    \hat{q}_{i'} \hat{\beta}_{i'}^{b-a}
\end{equation}
where $\approx$ denotes an approximation up to some additive error $\eta_2$.
To determine what $\eta_2$ should be, we analyze both sides of \cref{eq:coeff-approx} via the bounds from \cref{thm:matrix-pencil-sample-complexity}.

Let $L_{j-1}$ and $L_{j}$ respectively denote the number of samples that fall in the intervals $[a,\infty)$ and $[b,\infty)$.
For the RHS, note that if the additive errors $\epsilon_1$ and $\epsilon_2$ for $q_{i'}$ and $\beta_{i'}$ respectively are small, then the error of the right-hand side of \cref{eq:coeff-approx} is bounded by $\beta\epsilon_1 + q\epsilon_2 + \epsilon_1 \epsilon_2 \leq \epsilon_1 + \epsilon_2 + \epsilon_1\epsilon_2 = O\left( \left(\frac{2e}{\Delta}\right)^{3 D_{j-1}} \sqrt{\frac{1}{L_{j-1}}} \right)$.
The same analysis goes for the LHS. 
In particular, the error for $\Pr(T > b \;|\; T > a)$, by the DKW inequality, is $\sqrt{\frac{1}{L_{j-1}} \log\frac{2}{\delta_0}}$ with probability $\delta_0$, which is dominated by the error $O\left( \left( \frac{2e}{\Delta} \right)^{3D_j} \sqrt{\frac{1}{L_j}} \right)$ in $p_i$.
Therefore, the tolerance should be
$\eta_2 = \Theta\left( \left(\frac{2e}{\Delta}\right)^{3 D_j} \sqrt{\frac{1}{L_j}} \right)$.


\subsection{The algorithm description}
Now we implement the proposed strategy, which is built upon the ideas of the previous section (\cref{sec:reconstruction-lems}).

\vspace{5mm}
\noindent
\underline{\textbf{Inputs:}}
$\{c_1, \ldots, c_L\}$, a sample collection of $L$ iid coalescence times.

\vspace{5mm}
\noindent
\underline{\textbf{Parameters:}}
\begin{itemize}
	\item $D_0$, an initial upper bound for the number of subpopulations at $t=0$.
	\item $K$, the number of intervals for reconstruction.
	\item $\delta$, failure probability budget.
	\item $\epsilon$, a scaling parameter for the interval sizes.
	\item $N$, the present-day $(t=0)$ total population size.
	\item $P$, threshold lower bound for component weight.
	\item $\Delta$, threshold lower bound for exponential gap over all intervals.
\end{itemize}

\noindent
\underline{\textbf{Output:}}
A population history $\{ (I_j, D_j, \vN_j, \cE_j) \}_{j=1}^{K}$.

\vspace{3mm}
\noindent
\underline{\textbf{Procedure:}}
\begin{enumerate}[1.]
	\item Partition time (oriented towards the past) into intervals by setting $t_j = \epsilon j N$, for $j = 0, \ldots, K$. 
	Let $I_1 = [t_{0}, t_{1}]$, \ldots, $I_{K} = [t_{K-1}, t_{K}]$.
	Initialize $j = 1$.
	\item 
	\textbf{Collect CDF statistics in $I_j$}:
	Abbreviate $D := D_{j-1}$.
	For $\ell = 0, 1, \ldots, 2 D + 3$, compute the statistic
	\[
		\hat{v}_{\ell} 
		= 
		\frac{\#\{ i : c_i \geq t_{j-1} + \frac{\epsilon N}{2D + 4} \ell \}}{\#\{ i : c_i \geq t_{j-1} \}}
	\]
	\item \textbf{Learn the model parameters in $I_j$}: apply the Matrix Pencil Method on $D+2$ components using inputs $\hat{v}_{0}, \ldots, \hat{v}_{2D+3}$, which outputs eigenvalues $\hat{\alpha}_0 > \cdots > \hat{\alpha}_{D+1}$ and corresponding coefficients $\hat{p}_0, \hat{p}_1, \ldots, \hat{p}_{D+1}$.

	Discard all $\hat{\alpha}$'s (and corresponding $\hat{p}$'s) of absolute value at most $\eta$, where
	\begin{equation}\label{eq:error-eta1}
		\eta_1 
		= 
		\frac{4(D+2)^{3.5}}{P^2} \left( \frac{2e}{\Delta} \right)^{2D+4} \sqrt{\frac{1}{Le^{-\epsilon j N}} \log \frac{K}{\delta}}.
	\end{equation}
	Set $D_j$ equal to the number of leftover $\hat{\alpha}$'s, minus 1 (to account for the constant term), and re-normalize the remaining $\hat{p}$'s so that they sum to $1$.
	\item \textbf{Learn the population event at $t_{j-1}$ \& sizes in $I_j$}:
	If $j = 1$, let $\hat{\vN}_1 = (N \sqrt{ \hat{p}_1 }, \ldots, N \sqrt{ \hat{p}_{D_1} })$. (The leading coefficient $\hat{p}_{0}$ is omitted from this calculation.)
	
	Otherwise ($j > 1$), set $\hat{\vN}_j = (\frac{ \epsilon N }{ (2D+4) \hat{\mu}_1}, \ldots, \frac{ \epsilon N }{ (2D+4) \hat{\mu}_{D_j} })$, where $\hat{\mu}_{i} = -\log(\hat{\alpha}_i)$.
    Let $\{ (\hat{q}_{i}, \hat{\beta}_{i}) \}_{i=0}^{D_{j-1}}$ be the (coefficient, eigenvalue) pairs recovered from the previous iteration.
    
    Say that subpopulation $i$ of $I_{j}$ \emph{matches} subpopulation $i'$ of $I_{j-1}$ if
    \begin{equation}\label{eq:pops-match}
    	\left| 
        \hat{p}_i
    	-
    	\hat{q}_{i'} \hat{\beta}_{i'}^{2D+4} \cdot \frac{\#\{ k : t_k > a\}}{\#\{ k : t_k > b\}}
        \right| < \eta_2
    \end{equation}
    where 
    \begin{equation}\label{eq:error-eta2}
    	\eta_2 
        = 
        \frac{40 (D+2)^5}{P^2} 
        \left( \frac{2e}{\Delta}\right)^{3(D+2)} 
        \sqrt{\frac{1}{Le^{-\epsilon j N}} \log(K/\delta)}.
    \end{equation}
    
	Do one of the following.
	\begin{itemize}
		\item If $D_j = D + 1$, infer ``Split'':
		Find a valid matching.
        There should be one $\ell \in [D_{j-1}]$ and two indices $k_1, k_2 \in [D_j]$ left over.
		Assign the event: Split($\ell \to \{ k_1, k_2 \}$).
		\item If $D_j = D-1$, infer ``Merge'':
		Find a valid matching.
        There should be two $\ell_1, \ell_2 \in [D_{j-1}]$ and one $k \in [D_j]$ left over.
		Assign the event: Merge($\{\ell_1, \ell_2\} \to k$).
		\item If $D_j = D$, infer ``Change Size'':
		Identify a bijection from $[D_j]$ to $[D_{j-1}]$ via the matching scheme.
		\item If none of the above cases are true, or if there is no such prescribed matching, fail by default and terminate the procedure.
	\end{itemize}
	\item If $j = K$, stop. Otherwise, set $j \gets j+1$, go back to step 2.
\end{enumerate}

\subsection{Reconstruction Guarantees}
Consider the setting where the intervals are known and are equal in size (which decides $\epsilon$). Given enough samples, the algorithm for population reconstruction presented in Section~\ref{sec:population-history} succeeds in reconstructing the correct population history with high probability. 
Indeed, with a large enough number of samples, the conditional tail distribution $F(t) = \Pr(T > t \;|\; T \geq t_j)$ can be approximated to arbitrarily high precision with high probability. In turn, this means that the empirical (conditional) tail probabilities in $I_j$ that are collected in step (2) of the algorithm are arbitrarily close to the ground truth tail probabilities.

We choose $2(D+2)$ maximally spread out points in the interval $I_j$, and input the empirical tail probabilities at these points to the Matrix Pencil Method. 
The \emph{robustness} of the Matrix Pencil Method (Theorem~\ref{thm:matrix-pencil-sample-complexity}) guarantees that the estimated exponents $\hat{\mu}_i$ are sufficiently close to their true values $\mu_i$ with high probability. 
Furthermore, each true exponent $\mu_i$ is roughly the reciprocal of $N_i$, the size of the $i^{th}$ subpopulation in the interval $I_j$, and so approximating the $\mu_i$'s  gives us a good approximation of the subpopulation sizes in the interval $I_j$.

Note the scaling factor $\nu = \frac{\epsilon N}{(2D+4)}$ that appears in the algorithm, which comes from the change of variables $t = \nu \tau$.
The rescaling reflects a time units conversion from \textit{generations} to \textit{coalescent units}, which is convenient because the robustness guarantees are provided via interpolation at integer points $\tau = 0,1,2,$ etc.
The precise relation between the $\mu_i$s and $N_i$s in our algorithm is
\[ \mu_i = \frac{\epsilon N}{(2D+4) N_i}.\]

Steps 3 and 4 contain errors $\eta_1, \eta_2$ with correction terms $Le^{-\epsilon j N}$ and $K/\delta$.
The first term $Le^{-\epsilon j N}$ is the expected number of samples (out of $L$ total) found inside and to the right of the interval $I_j$.
By the observations made in \cref{sec:reconstruction-lems} and \cref{sec:reconstruction-recursion}, Steps 3 and 4 of the algorithm provide accurate reconstructions of $\vN$ for $I_j$ with probability $\delta/K$, given sufficiently large samples.
By union bounding over all $K$ intervals, the algorithm accurately reconstructs all $K$ pieces of $\vN$ with probability $\delta$.


\begin{remark}
It is almost always the case, however, that we do not know the interval endpoints, and do not know a small enough value of $\epsilon$ for which the algorithm will work.
According to the model, there exists a sufficiently small value of $\epsilon$ which will allow for correct inference, where "sufficiently small" means $I_1, \ldots, I_K$ in Step 1 of the algorithm is fine-grained enough to capture all intervals.
Therefore, a natural strategy is to try a search-based approach for an optimal value, as follows.
Choose various values of $\epsilon$, then learn the parameters of the model for each.
From these, we may choose the best $\epsilon$ via a goodness-of-fit test, by taking the $\epsilon$ that outputs a learned distribution $\hat{f}$ that is closest to the empirical distribution $f_{emp}$ in total variation distance.
Note that $f_{emp}$ converges to the true $f$ as the number of samples increases, for some appropriate $\epsilon$.
\end{remark}

%% file: main.bbl
\begin{thebibliography}{10}

\bibitem{bhaskar2014descartes}
Anand Bhaskar and Yun~S Song.
\newblock Descartes’ rule of signs and the identifiability of population
  demographic models from genomic variation data.
\newblock {\em Annals of statistics}, 42(6):2469, 2014.

\bibitem{bhaskar2015efficient}
Anand Bhaskar, YX~Rachel Wang, and Yun~S Song.
\newblock Efficient inference of population size histories and locus-specific
  mutation rates from large-sample genomic variation data.
\newblock {\em Genome research}, pages gr--178756, 2015.

\bibitem{blythe2007stochastic}
Richard~A Blythe and Alan~J McKane.
\newblock Stochastic models of evolution in genetics, ecology and linguistics.
\newblock {\em Journal of Statistical Mechanics: Theory and Experiment},
  2007(07):P07018, 2007.

\bibitem{candes2013super}
Emmanuel~J Cand{\`e}s and Carlos Fernandez-Granda.
\newblock Super-resolution from noisy data.
\newblock {\em Journal of Fourier Analysis and Applications}, 19(6):1229--1254,
  2013.

\bibitem{Drummond2005}
AJ~Drummond, A~Rambaut, B~Shapiro, and OG~Pybus.
\newblock Bayesian coalescent inference of past population dynamics from
  molecular sequences.
\newblock {\em Mol. Biol. Evol.}, 22(5):1185--1192, 2005.

\bibitem{dvoretzky1956asymptotic}
Aryeh Dvoretzky, Jack Kiefer, and Jacob Wolfowitz.
\newblock Asymptotic minimax character of the sample distribution function and
  of the classical multinomial estimator.
\newblock {\em The Annals of Mathematical Statistics}, pages 642--669, 1956.

\bibitem{Excoffier2013}
Laurent Excoffier, Isabelle Dupanloup, Emilia Huerta-S{\'a}nchez, Vitor~C
  Sousa, and Matthieu Foll.
\newblock {Robust Demographic Inference from Genomic and SNP Data}.
\newblock {\em {PLoS Genetics}}, 9(10):e1003905, 2013.

\bibitem{feldmann1998fitting}
Anja Feldmann and Ward Whitt.
\newblock Fitting mixtures of exponentials to long-tail distributions to
  analyze network performance models.
\newblock {\em Performance evaluation}, 31(3-4):245--279, 1998.

\bibitem{gautschi1962inverses}
Walter Gautschi.
\newblock On inverses of vandermonde and confluent vandermonde matrices.
\newblock {\em Numerische Mathematik}, 4(1):117--123, 1962.

\bibitem{gautschi1990stable}
Walter Gautschi.
\newblock How (un) stable are vandermonde systems.
\newblock 1990.

\bibitem{heled2008bayesian}
Joseph Heled and Alexei Drummond.
\newblock Bayesian inference of population size history from multiple loci.
\newblock {\em BMC Evolutionary Biology}, 8(1):289, 2008.

\bibitem{MPM}
Yingbo Hua and Tapan~K Sarkar.
\newblock Matrix pencil method for estimating parameters of exponentially
  damped/undamped sinusoids in noise.
\newblock {\em IEEE Transactions on Acoustics, Speech, and Signal Processing},
  38(5):814--824, 1990.

\bibitem{joseph2018inference}
Tyler~A Joseph and Itsik Pe’er.
\newblock Inference of population structure from ancient dna.
\newblock In {\em International Conference on Research in Computational
  Molecular Biology}, pages 90--104. Springer, 2018.

\bibitem{kim2015can}
Junhyong Kim, Elchanan Mossel, Mikl{\'o}s~Z R{\'a}cz, and Nathan Ross.
\newblock Can one hear the shape of a population history?
\newblock {\em Theoretical population biology}, 100:26--38, 2015.

\bibitem{li2011inference}
Heng Li and Richard Durbin.
\newblock Inference of human population history from individual whole-genome
  sequences.
\newblock {\em Nature}, 475(7357):493, 2011.

\bibitem{massart1990tight}
Pascal Massart.
\newblock The tight constant in the dvoretzky-kiefer-wolfowitz inequality.
\newblock {\em The Annals of Probability}, pages 1269--1283, 1990.

\bibitem{mcvean2005approximating}
Gilean~AT McVean and Niall~J Cardin.
\newblock Approximating the coalescent with recombination.
\newblock {\em Philosophical Transactions of the Royal Society of London B:
  Biological Sciences}, 360(1459):1387--1393, 2005.

\bibitem{moitra14}
Ankur Moitra.
\newblock Super-resolution, extremal functions and the condition number of
  vandermonde matrices.
\newblock In {\em Proceedings of the Forty-seventh Annual ACM Symposium on
  Theory of Computing}, STOC '15, pages 821--830, New York, NY, USA, 2015. ACM.

\bibitem{myers2008can}
Simon Myers, Charles Fefferman, and Nick Patterson.
\newblock Can one learn history from the allelic spectrum?
\newblock {\em Theoretical population biology}, 73(3):342--348, 2008.

\bibitem{nazarov1993local}
Fedor~L'vovich Nazarov.
\newblock Local estimates for exponential polynomials and their applications to
  inequalities of the uncertainty principle type.
\newblock {\em Algebra i analiz}, 5(4):3--66, 1993.

\bibitem{Nielsen2000}
Rasmus Nielsen.
\newblock Estimation of population parameters and recombination rates from
  single nucleotide polymorphisms.
\newblock {\em Genetics}, 154(2):931--942, 2000.

\bibitem{nordborg2001coalescent}
Magnus Nordborg.
\newblock Coalescent theory.
\newblock {\em Handbook of statistical genetics}, 2:843--877, 2001.

\bibitem{prony}
R~Prony.
\newblock Essai éxperimental et analytique: sur les lois de la dilatabilité
  de uides élastique et sur celles de la force expansive de la vapeur de
  l'alkool, à diérentes températures.
\newblock {\em Journal de l’Ecole Polytechnique}, 2, 1795.

\bibitem{remez1936propriete}
EJ~Remez.
\newblock Sur une propri{\'e}t{\'e} des polyn{\^o}mes de tchebycheff, comm.
  l’inst.
\newblock {\em Sci., Kharkow}, 13:93--95, 1936.

\bibitem{schiffels2014inferring}
Stephan Schiffels and Richard Durbin.
\newblock Inferring human population size and separation history from multiple
  genome sequences.
\newblock {\em Nature genetics}, 46(8):919, 2014.

\bibitem{Sheehan2013}
Sara Sheehan, Kelley Harris, and Yuns~S. Song.
\newblock {Estimating Variable Effective Population Sizes from Multiple
  Genomes: A Sequentially Markov Conditional Sampling Distribution Approach}.
\newblock {\em Genetics}, 194:647--662, 2013.

\bibitem{terhorst2017robust}
Jonathan Terhorst, John~A Kamm, and Yun~S Song.
\newblock Robust and scalable inference of population history from hundreds of
  unphased whole genomes.
\newblock {\em Nature genetics}, 49(2):303, 2017.

\bibitem{terhorst2015fundamental}
Jonathan Terhorst and Yun~S Song.
\newblock Fundamental limits on the accuracy of demographic inference based on
  the sample frequency spectrum.
\newblock {\em Proceedings of the National Academy of Sciences},
  112(25):7677--7682, 2015.

\bibitem{turan1984new}
Paul Tur{\'a}n.
\newblock {\em On a new method of analysis and its applications}.
\newblock Wiley New York, 1984.

\end{thebibliography}
